\documentclass[aps,prx,onecolumn,superscriptaddress,nofootinbib,longbibliography]{revtex4-2}
\usepackage[a4paper, left=1in, right=1in, top=0.95in, bottom=0.95in]{geometry}
\usepackage{fix-cm} 

\usepackage{cmap} 
\usepackage[utf8]{inputenc}
\usepackage[english]{babel}
\usepackage[T1]{fontenc}
\usepackage{setspace}
\usepackage{amsmath}
\usepackage{amssymb}
\usepackage{amsfonts}
\usepackage{amsthm}
\usepackage{mathtools}
\usepackage{thmtools,thm-restate}
\usepackage{physics}
\usepackage{dsfont}
\usepackage{wasysym}
\usepackage{bm}
\usepackage{bbm}
\usepackage{etoolbox}
\usepackage{subcaption}
\usepackage{xcolor}
\definecolor{blueviolet}{rgb}{0.2, 0.2, 0.6}
\definecolor{webgreen}{rgb}{0,.5,0}
\definecolor{webbrown}{rgb}{.6,0,0}
\usepackage[pdftex,
	bookmarks=false,
	colorlinks=true, 
	urlcolor=webbrown, 
	linkcolor=blueviolet, 
	citecolor=webgreen,
	pdfstartpage=1,
	pdfstartview={FitH},  
	bookmarksopen=false
	]{hyperref}
\allowdisplaybreaks
\usepackage{braket}
\usepackage{natbib}
\usepackage{listings}
\usepackage[capitalize]{cleveref}

\usepackage{graphicx}
\usepackage[normalem]{ulem}
\usepackage{wrapfig}
\usepackage[export]{adjustbox}

\usepackage{caption}
\numberwithin{equation}{section}

\newtheorem{theorem}{Theorem}
\newtheorem{corollary}{Corollary}

\newtheorem{definition}{Definition}

\newtheorem{lemma}{Lemma}
\newtheorem{proposition}{Proposition}
\newtheorem{fact}{Fact}

\theoremstyle{definition}

\newcommand{\nocontentsline}[3]{}
\let\origcontentsline\addcontentsline
\newcommand\stoptoc{\let\addcontentsline\nocontentsline}
\newcommand\resumetoc{\let\addcontentsline\origcontentsline}

\newcommand{\rom}[1]{\mathtt{\uppercase\expandafter{\romannumeral #1\relax}}}

\DeclareMathOperator*{\E}{{\mathbb{E}}}

\DeclareMathOperator{\poly}{poly}

\providecommand{\bij}{\mathsf{bij}}

\newcommand{\eps}{\varepsilon}

\newtheorem*{theorem*}{Theorem}

\newtheorem*{task*}{Task}
\newtheorem*{proposition*}{Proposition}

\providecommand{\scC}{\mathsf{cC}}
\providecommand{\scD}{\mathsf{cD}}

\newcommand{\init}{\mathsf{init}}

\newcommand{\ee}{\end{equation}}

\newcommand{\calO}{\mathcal{O}}

\providecommand{\Id}{\mathbbm{1}}

\providecommand{\num}{\mathsf{num}}

\providecommand{\ssym}{\mathsf{sym}}

\providecommand{\ssym}{\mathsf{sym}}

\providecommand{\Haar}{\mathsf{Haar}}

\providecommand{\sA}{\mathsf{A}}
\providecommand{\sB}{\mathsf{B}}

\providecommand{\sL}{\mathsf{L}}

\providecommand{\sR}{\mathsf{R}}

\DeclareMathOperator{\Dom}{Dom}

\providecommand{\gsA}{{\textcolor{gray}{\mathsf{A}}}}
\providecommand{\gsA}{{\textcolor{gray}{\mathsf{A}_{\mathsf{abc}}}}}
\providecommand{\gsB}{{\textcolor{gray}{\mathsf{B}}}}
\providecommand{\gsC}{{\textcolor{gray}{\mathsf{C}}}}
\providecommand{\gsD}{{\textcolor{gray}{\mathsf{D}}}}
\providecommand{\gsE}{{\textcolor{gray}{\mathsf{E}}}}

\providecommand{\gsL}{{\textcolor{gray}{\mathsf{L}}}}

\providecommand{\gsR}{{\textcolor{gray}{\mathsf{R}}}}

\providecommand{\calD}{\mathcal{D}}
\providecommand{\calE}{\mathcal{E}}

\providecommand{\calO}{\mathcal{O}}

\providecommand{\calR}{\mathcal{R}}
\providecommand{\calS}{\mathcal{S}}
\providecommand{\calT}{\mathcal{T}}

\providecommand{\calW}{\mathcal{W}}

\DeclareMathOperator{\dist}{{dist}}

\crefformat{footnote}{#2\footnotemark[#1]#3}

\AfterEndEnvironment{property}{\noindent\ignorespaces}

\newcommand{\expect}{\mathop{\mathbb{E}}}

\begin{document}
\fontsize{10}{12}\selectfont

\title{Random unitaries from Hamiltonian dynamics}

\author{Laura Cui}
\affiliation{California Institute of Technology, Pasadena, California 91125, USA}

\author{Thomas Schuster}
\affiliation{California Institute of Technology, Pasadena, California 91125, USA}
\affiliation{Google Quantum AI, Venice, California 90291, USA}

\author{Liang Mao}
\affiliation{California Institute of Technology, Pasadena, California 91125, USA}
\affiliation{Institute for Advanced Study, Tsinghua University, Beijing, 100084, China}

\author{Hsin-Yuan Huang}
\affiliation{California Institute of Technology, Pasadena, California 91125, USA}
\affiliation{Google Quantum AI, Venice, California 90291, USA}

\author{Fernando Brand\~ao}
\affiliation{AWS Center for Quantum Computing, Pasadena, California 91125, USA}
\affiliation{California Institute of Technology, Pasadena, California 91125, USA}

\date{\today}

\begin{abstract}
The nature of randomness and complexity growth in systems governed by unitary dynamics is a fundamental question in quantum many-body physics. This problem has motivated the study of models such as local random circuits and their convergence to Haar-random unitaries in the long-time limit. However, these models do not correspond to any family of physical time-independent Hamiltonians. In this work, we address this gap by studying the indistinguishability of time-independent Hamiltonian dynamics from truly random unitaries. On one hand, we establish a no-go result showing that for any ensemble of constant-local Hamiltonians and any evolution times, the resulting time-evolution unitary can be efficiently distinguished from Haar-random and fails to form a $2$-design or a pseudorandom unitary (PRU). On the other hand, we prove that this limitation can be overcome by increasing the locality slightly: there exist ensembles of random polylog-local Hamiltonians in one-dimension such that under constant evolution time, the resulting time-evolution unitary is indistinguishable from Haar-random, i.e.~it forms both a unitary $k$-design and a PRU. Moreover, these Hamiltonians can be efficiently simulated under standard cryptographic assumptions. 


\end{abstract}

\maketitle

\stoptoc
\section{Introduction}

Characterizing the emergence of universal chaotic and ergodic behaviors is a central goal of quantum many-body physics.
Seminal early works, for example, on spectral distributions and single-particle quantum chaos~\cite{wigner1967random,bohigas1984characterization,muller2004semiclassical}, and the eigenstate thermalization hypothesis~\cite{deutsch1991quantum,rigol2008thermalization}, have shown that small amounts of uncertainty in the parameters of a chaotic Hamiltonian can lead to universal fluctuations in its properties and observables.
%
In the last decade, tremendous further progress has been made in \emph{time-dependent} quantum systems, such as random quantum circuits, by capturing the emergence of universal chaotic behaviors using the notions of \emph{unitary $k$-designs}~\cite{emerson2003pseudo,gross2007evenly,dankert2005efficient,dankert2009exact,brandao2016local,nakata2016efficient,haah2024efficient, chen2024incompressibility,guo2024complexity,laracuente2024approximate,schuster2024random,schuster2024random,lami2025anticoncentration,grevink2025will,cui2025unitary,schuster2025strong,foxman2025random,schuster2025hardness,zhang2025designs} and \emph{pseudorandom unitaries} (PRUs)~\cite{ji2018pseudorandom,brakerski2019pseudo,metger2024simple,chen2024efficient,ma2025construct,schuster2024random,schuster2025strong,ananth2025pseudorandom}.
These objects capture how quickly the properties of a physical system can become indistinguishable from those of a completely Haar-random unitary.
Due to their elegance and broad range of physical predictions, unitary designs and PRUs have been widely applied in quantum benchmarking~\cite{emerson2005scalable,knill2008randomized,elben2023randomized,guta2020fast, huang2020predicting,zhao2021fermionic}, demonstrations of quantum advantage~\cite{arute2019quantum, morvan2023phase, abanin2025constructive}, and even fundamental questions in quantum gravity~\cite{hayden2007black,bouland2019computational,kim2020ghost,akers2022black,akers2024holographic,yang2025complexity}. More generally, they have also been employed to understand properties of physical dynamics through the lens of learnability~\cite{aharonov2021quantum,huang2021quantum,cotler2023information,schuster2024random}, information encoding and scrambling~\cite{BF13,yoshida2017efficient,landsman2019verified,blok2021quantum,BF12,nahum2017entgrowth,nahum2018operator,schuster2022many,schuster2025strong}, sampling complexity~\cite{boixo2018characterizing,bouland2019complexity,lami2025anticoncentration},  
computational power~\cite{brown2018second,haferkamp2022linear,chen2024incompressibility}, and the onset of quantum chaos~\cite{maldacena2016bound,roberts2017chaos,cotler2017chaos,pilatowsky2024hilbert,gu2024simulating,pilatowsky2024hilbert}.


Despite this success, random quantum circuits are imperfect models of physical systems in important ways.
Most prominently, they are time-dependent, and hence do not conserve energy, as any physical Hamiltonian time-evolution would.
Useful attempts to bridge this gap have focused on \emph{charge}-conserving random  circuits~\cite{kong2021charge,kong2022near,li2024efficient,li2024designs,liu2024unitary,mitsuhashi2025unitary,haah2025short}; however, such systems still miss many key features of physical Hamiltonian evolution.
For example, the conserved charge is usually especially simple and fixed; the conserved charge also does not govern the time-evolution of the system, as a physical Hamiltonian would.
%
Hence, it remains an open question to what extent unitary designs and PRUs can in fact capture the chaotic behaviors of physical time-independent Hamiltonian dynamics.
This question is fundamentally important, as most applications of random unitaries in many-body and high-energy physics are primally concerned with time-independent Hamiltonian systems.

In this work, we establish several central results on the formation of unitary designs and PRUs in time-independent Hamiltonian dynamics.
Our first result is a broad no-go theorem.
We provide a simple and efficient test to distinguish any ensemble of time-evolutions under any constant-local Hamiltonians for any lengths of time from a Haar-random unitary, using only two queries to the time-evolution and minimal classical computation.
This proves that time-evolution under constant-local Hamiltonians can never form PRUs, or unitary 2-designs with polynomially small error.
Motivated by this fact, our second result considers the behavior of Hamiltonians with a slightly higher locality.
In sharp contrast to our results on constant-local Hamiltonians, we prove that ensembles of time-evolution under random Hamiltonians of locality $\omega(\log n)$ can readily form both unitary $k$-designs and PRUs.
This holds even for constant evolution times, $t= \mathcal{O}(1)$.
We also prove several additional results on non-adaptive PRUs, which 
broaden the range of Hamiltonian ensembles to which our results apply.


Our results establish a need for nuance when applying random unitary descriptions to physical Hamiltonian dynamics.
In the long-term, this emphasizes the need for new and more precise formulations of universal chaotic behaviors that can capture the physics of local Hamiltonian systems.
Recent works suggest several promising approaches in this direction~\cite{pappalardi2022eigenstate,fava2025designs,cotler2023emergent,mark2024maximum}.
Our results on the formation of unitary designs and PRUs also complement other recent discoveries on the robust emergence of ergodic phenomena in Hamiltonians with a slightly increased locality~\cite{anschuetz2025strongly}.
From an applications perspective, our findings immediately imply fundamental limits on learning properties from time-independent Hamiltonian dynamics. 
For example, the complexity of learning the parameters of an unknown Hamiltonian~\cite{wiebe2014hamiltonian,evans2019scalable,haah2021optimal,huang2023learningb,dutkiewicz2023advantage,bakshi2024structure} must grow exponentially in the Hamiltonian locality, even in situations where the Hamiltonian itself possesses an efficient description.

\begin{figure}
    \centering
    \includegraphics[width=0.42\linewidth]{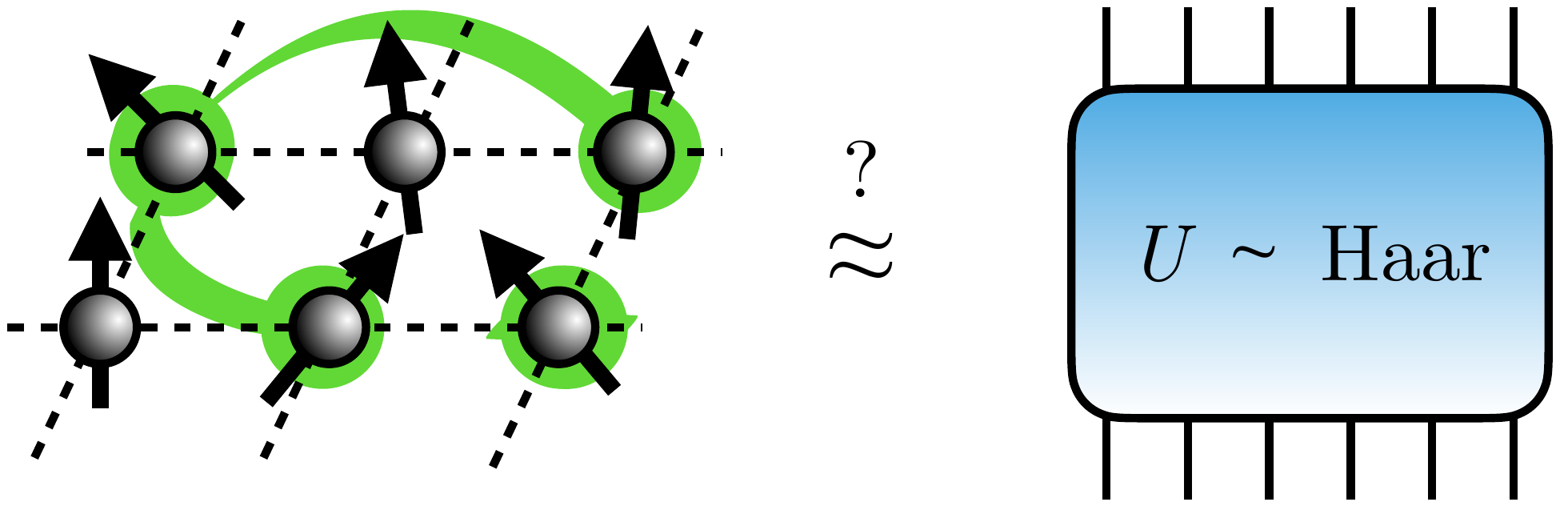}
    \caption{Our work investigates when time-independent Hamiltonian dynamics are asymptotically indistinguishable from random unitary transformations. We find that pseudorandomness cannot arise in any family of Hamiltonians with constant-local interactions, but can be achieved in quasi-local models.}
    \label{fig:dynamics random unitary}
\end{figure}


\subsection*{Relation to previous work}

Despite considerate interest in the use of random unitaries as models of physical systems~\cite{wigner1967random,bohigas1984characterization,muller2004semiclassical,deutsch1991quantum,rigol2008thermalization,hayden2007black,bouland2019computational,kim2020ghost,akers2022black,akers2024holographic,yang2025complexity,maldacena2016bound,roberts2017chaos,cotler2017chaos,pilatowsky2024hilbert,gu2024simulating,pappalardi2022eigenstate,fava2025designs,cotler2023emergent,mark2024maximum,kong2021charge,kong2022near,li2024efficient,li2024designs,liu2024unitary,mitsuhashi2025unitary,haah2025short,pilatowsky2024hilbert}, there are relatively few existing results on the realization of unitary designs and PRUs from time-independent Hamiltonian evolution.
This stems from the notorious difficulty of rigorously analyzing the dynamical properties of many-body quantum Hamiltonians.
Several works have considered time-\emph{dependent} Hamiltonians~\cite{nakata2016efficient,guo2024complexity}.
However, such models do not conserve energy, and share more features in common with random quantum circuits than with time-independent Hamiltonian systems. 
%
More relevantly, recent work has shown that time-evolution under a random time-\emph{independent} Hamiltonian, drawn from either the GUE ensemble on $n$ qubits or a pseudorandom variant of it, can become indistinguishable from a Haar-random unitary after a super-polynomially long evolution time, $t = \omega(\poly n)$~\cite{gu2024simulating}.
Our work substantially expands upon this result, by showing that time-evolution under nearly-local Hamiltonians, of locality $\omega(\log n)$ instead of $\mathcal{O}(n)$, can become indistinguishable from a Haar-random unitary after only very short evolution times, $t = \mathcal{O}(1)$ instead of $t = \omega(\poly n)$.
Our no-go theorems also highlight the precise limits of this indistinguishability for Hamiltonians of any smaller locality.






\section{Preliminaries}

Before turning to our results, we provide a short review of the definitions of unitary $k$-designs and pseudorandom unitaries (PRUs).
As aforementioned, both notions seek to capture the closeness of a random unitary ensemble $\mathcal{E}$ to the Haar ensemble, under various metrics of approximation.
A Haar-random unitary on $n$ qubits is a random unitary matrix drawn from from the Haar measure on $U(2^n)$.

An approximate unitary $k$-design is a unitary ensemble $\mathcal{E}$ that reproduces the $k$-th moment of the Haar ensemble up to a small  error.
In particular, an ensemble $\mathcal{E}$ is an approximate unitary $k$-design up to \emph{additive error} $\varepsilon$ if its $k$-th moment, $\Phi_\mathcal{E}(\cdot) \equiv U^{\otimes k} (\cdot) U^{\dagger, \otimes k}$, is close to the $k$-th moment $\Phi_H(\cdot)$ of the Haar ensemble in the diamond norm, $\lVert \Phi_\mathcal{E} - \Phi_H \rVert_{\diamond} \equiv \max_\rho \lVert \Phi_\mathcal{E}(\rho) - \Phi_H \rVert_1 \leq \varepsilon$~\cite{brandao2016local}.
This bounds the distinguishability of a random unitary drawn from $\mathcal{E}$ from a random unitary drawn from the Haar ensemble, by any quantum experiment that queries the random unitary $k$ times in parallel~\cite{cui2025unitary}.
A more recent and much stronger notion of approximation error for unitary designs is the measurable error~\cite{cui2025unitary}.
An ensemble $\mathcal{E}$ is an approximate unitary $k$-design up to \emph{measurable error} $\varepsilon$ if the expected output state of any quantum experiment that queries a random unitary drawn from $\mathcal{E}$ is close to the expected output state for the same experiment querying a Haar-random unitary up to small trace-norm error.
This bounds the distinguishability in any quantum experiment that queries the random unitary $k$ times and performs arbitrary quantum operations and measurements in between~\cite{cui2025unitary}.
Unless otherwise specified, we will utilize this strong form of measurable approximation error throughout our work.

A pseudorandom unitary (PRU) is a unitary ensemble $\mathcal{E}$ that is indistinguishable from the Haar ensemble in any \emph{efficient} quantum experiment.
In particular, an ensemble $\mathcal{E}$ is a PRU ensemble with security against any $t(n)$-time quantum adversary if it cannot be distinguished from Haar-random by any quantum experiment that runs in $t(n)$ time, where $t(n)$ is any function of $n$~\cite{ji2018pseudorandom,ma2025construct}.
Unless otherwise stated, in our work we will consider security against any polynomial-time quantum experiment, $t(n) = \poly n$.
Weaker forms of security can also be considered.
An ensemble $\mathcal{E}$ is a PRU ensemble with security against any \emph{non-adaptive} $t(n)$-time quantum adversary if it cannot be distinguished from Haar-random by any quantum experiment that runs in $t(n)$ time and queries many applications of the unitary $U$ all at once in parallel~\cite{metger2024simple,chen2024efficient}.
The standard (i.e.~adaptive) security of PRUs is analogous to the measurable error of unitary designs, and the non-adaptive security of PRUs is analogous to the additive error of unitary designs.

The standard definitions of approximate unitary $k$-designs and PRUs only capture quantum experiments that perform \emph{forward} queries to a random unitary, i.e.~those that query $U$~\cite{schuster2024random}.
Experiments that perform \emph{backward} queries, i.e.~which query both $U$ and its inverse $U^\dagger$, are not captured by these definitions~\cite{cotler2023information,schuster2024random}.
While forward-only access is sufficient for most applications of designs and PRUs, in some cases it can be desirable to capture queries to the inverse unitary as well.
To this end, one can define \emph{strong approximate unitary $k$-designs}~\cite{schuster2025strong} and \emph{strong PRUs}~\cite{ma2025construct,schuster2025strong} in an identical manner to standard designs and PRUs, but now allowing queries to the inverse (and conjugate and transpose) as well. 

\section{Impossibility of random unitaries in constant-local Hamiltonians}

We begin by addressing the most natural physical context: time-evolution under constant-local Hamiltonians.
We consider any unitary ensemble composed of time-evolution operators under any $q$-local Hamiltonians for any evolution times, 
\begin{equation}
    \mathcal{E} = \{ e^{-iHt} \, | \, (H, t) \sim \mathcal{D} \},
\end{equation}
where $\mathcal{D}$ is an arbitrary distribution over $q$-local Hamiltonians $H$ and times $t$.
For example, these include the temporal ensemble~\cite{mark2024maximum}, where one fixes the Hamiltonian $H$ and randomizes the time $t \sim [0,\infty)$.
They also include examples in which the time $t$ is fixed and the Hamiltonian itself is randomized~\cite{deutsch1991quantum}, as well as examples where both $H$ and $t$ are jointly distributed.

When the Hamiltonian $H$ is fixed, i.e.~only the time is randomized in $\mathcal{D}$, it is obvious that the ensemble $\mathcal{E}$ cannot form a unitary 1-design or PRU.
This follows because one can prepare an initial state with non-zero energy under $H$ (assuming without loss of generality that $\Tr(H) = 0$), apply a unitary drawn from $\mathcal{E}$, and then measure the final energy of the system.
If the unitary were Haar-random, the energy would equilibrate to zero with high probability.
In contrast, under any time-evolution by $H$, the energy is conserved and so remains equal to its initial value.

A more interesting setting is when the Hamiltonian $H$ is unknown a priori (i.e.~random).
This foils the naive distinguishing strategy above, since one does not known which basis to perform the initial state preparation and final measurement in.
Nonetheless, our main result on constant-local Hamiltonians shows that this difficulty can be surmounted.
We provide a simple and efficient test to distinguish any ensemble of $q$-local Hamiltonian evolutions from Haar-random using only two queries, product state preparation and read-out, and minimal classical computation.


%
This leads immediately to our no-go theorems on unitary $k$-designs (for any $k \geq 2$) and PRUs:
\begin{theorem}
    [Time-evolution under constant-local Hamiltonians cannot form unitary designs] \label{thm:nogodesign}
    The ensemble $\mathcal{E}$ cannot form a unitary 2-design for any additive error $\varepsilon \leq \mathcal{O}(1/12^q n)$ in one-dimensional systems, nor for any $\varepsilon \leq \mathcal{O}(1/12^q n^q)$ in all-to-all connected systems.
\end{theorem}
\begin{theorem}
    [Time-evolution under constant-local Hamiltonians cannot form PRUs] \label{thm:nogoPRU}
    The ensemble $\mathcal{E}$ cannot form a PRU for any $q = \mathcal{O}(\log n)$ in one-dimensional systems, nor for any $q = \mathcal{O}(1)$ in all-to-all connected systems.
\end{theorem}
\noindent Both no-go theorems apply to any constant-local Hamiltonian time-evolution ensemble, even those with arbitrarily long evolution times.
The $q$-dependence of both theorems is optimal, following our results in the following section.
The proof of the theorems is summarized below and provided in detail in Appendix~\ref{app:nogo}.


Our distinguishing protocol is exceptionally simple, and bears a conceptual similarity to the distinguishing protocol discussed above when the Hamiltonian is fixed and known.
We draw a random stabilizer product state $\ket{u}$ and a random Pauli operator $P_i$ from the set of Pauli operators allowed to appear in $H$.
That is, if $H$ is one-dimensional, we draw $P_i$ from the set of all $q$-geometrically-local Pauli operators; there are $\mathcal{O}(4^q n)$ such operators.
If $H$ is all-to-all-connected, we draw $P_i$ from the set of all Pauli operators with weight less than or equal to $q$; there are $\mathcal{O}((4n)^q)$ such operators.
We then prepare two copies of the product state $\ket{u}$, perform time-evolution on both copies under the same random unitary drawn from $\mathcal{E}$, and measure the Pauli operator $P_i$ on both copies.
Our main result is a proof that the expected value of the two-copy operator, $P_i \otimes P_i$---i.e.~the probability that the measurements on both copies return the same value---is greater than a constant threshold value for any ensemble of constant-local Hamiltonian time-evolutions.
This allows one to distinguish constant-local Hamiltonian evolutions from Haar-random, since the same expectation value is near zero under Haar-random evolution.

Intuitively, this protocol succeeds because a random product state will have a non-negligible energy under the Hamiltonian with high probability whenever $H$ is constant-local.
This follows because each term in the Hamiltonian has non-zero expectation value in a random product basis with probability $\mathcal{O}(1/3^q)$.
After time-evolution, this non-negigible energy of the state must be stored in the expectation values of \emph{some} set of Pauli operators appearing in the Hamiltonian.
By selecting a random Pauli operator in the Hamiltonian, we can detect this energy with high probability, and distinguish the unitary time-evolution from Haar-random.
The precise dependence on $n$ and $q$ follows by multiplying the probability, $\mathcal{O}(1/3^q)$, with the inverse number of Pauli operators that can appear in the Hamiltonian (Appendix~\ref{app:nogo}).

Finally, in the following section, we provide an additional lower bound (Proposition~\ref{prop:temp}) which shows that any Hamiltonian time-evolution ensemble with polynomial evolution time requires $\mathcal{O}(nk)$ bits of randomness in the Hamiltonian.
This requires Hamiltonians of locality $q = \Omega(\log k)$ in one-dimensional systems, and $q = \Omega(\log k / \log n)$ in all-to-all connected systems.

\begin{figure}
    \centering
    \includegraphics[width=0.88\linewidth]{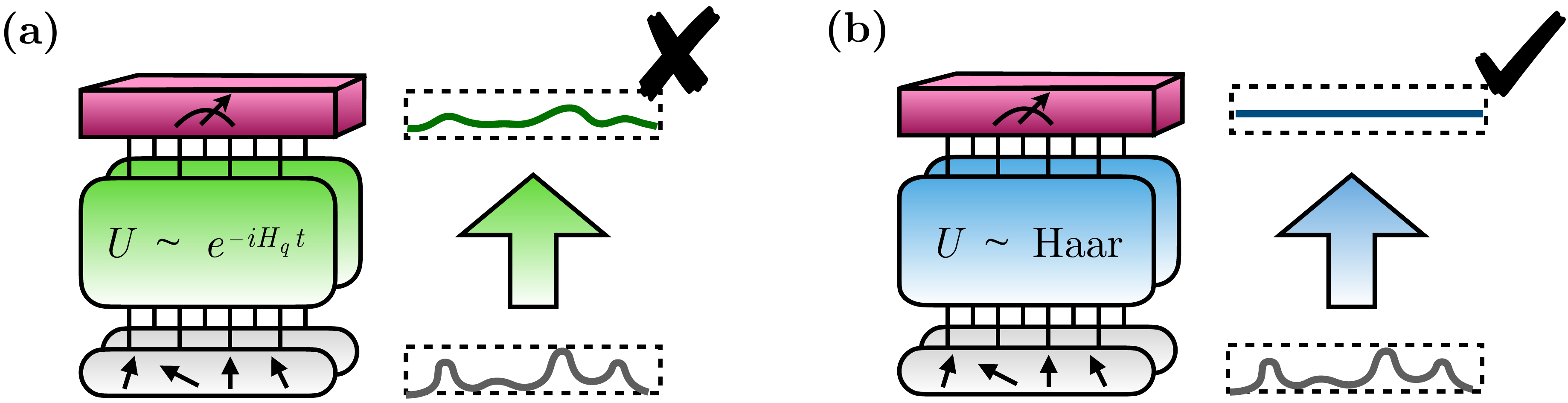}
    \caption{Illustration of our no-go result for constant-local dynamics. We show that it is possible to efficiently distinguish any ensemble $\calE$ generated by evolving under $q$-local Hamiltonians up to arbitrarily long times from a Haar-random unitary by measuring random Pauli operators of weight up to $q$ on two copies of the system, as \textbf{(a)} the output distribution of $\calE$ retains local correlations, while \textbf{(b)} the output distribution of the Haar ensemble appears completely uniform.}
    \label{fig:q-local vs haar}
\end{figure}


\section{Formation of random unitaries in nearly-local Hamiltonians}

Motivated by our no-go theorems for constant-local Hamiltonian time-evolution, we will now relax our requirements, and consider Hamiltonians whose locality increases very slowly with the system size, $q = \poly(\log n)$.
These Hamiltonians are not commonly realized in nature, where almost all systems are constant-local.
Nevertheless, the study of such  Hamiltonians with an increased locality has a long and fruitful history in high-energy, condensed matter, and mathematical physics~\cite{mezard2009information,steinacker2010emergent,maldacena2016remarks,qi2019quantum,lin2022infinite,berkooz2025cordial,swingle2024bosonic,guo2024complexity,anschuetz2025strongly,chen2024sparse}.
In principle, such Hamiltonians might emerge as effective models for systems restricted to a low-energy subspace of their full Hilbert space. 
From a theoretical perspective, they will also allow us to establish that the exponential dependence on the Hamiltonian locality $q$ in our no-go theorems is fundamental.

Our main result shows that short time-evolutions under Hamiltonians with logarithmic locality, $q = \mathcal{O}(\log n)$, are capable of forming $\varepsilon$-approximate unitary $k$-designs for any $\varepsilon = 1/\poly n$.
Meanwhile, short-time evolutions under Hamiltonian with any super-logarithmic locality, $q = \omega(\log n)$, are capable of forming PRUs.
This sharply contrasts with the behavior of Hamiltonians of lower locality, which are not capable of realizing designs or PRUs even after arbitrarily long evolution times, from Theorems~\ref{thm:nogodesign} and~\ref{thm:nogoPRU}.
In order to rigorously establish the formation of designs and PRUs, we consider the following random 1D Hamiltonian ensemble,
\begin{equation}\nonumber
    H = \left( \otimes_{i \in \text{even}} U_{i,i+1} \right)^\dagger  \left( \otimes_{i \in \text{odd}} U_{i,i+1} \right)^\dagger  \left( \sum_{i=1}^{n/\xi} \sum_{z \in \{0,1\}^{\xi}} \!\!\!\! J^i_z \dyad{z}_{i} \right)  \left( \otimes_{i \in \text{odd}} U_{i,i+1} \right)  \left( \otimes_{i \in \text{even}} U_{i,i+1} \right).
\end{equation}
Here, we divide $n$ qubits along a 1D line into $n/\xi$ patches of $\xi$ qubits each.
We label each patch with the index $i = 1,\ldots,n/\xi$.
The spectrum of the Hamiltonian is given by the central term, which is a sum of random energies $J^i_z \in [-1,1]$ in the computational basis on each patch.
The eigenbasis of the Hamiltonian is scrambled by the remaining unitary terms, which conjugate the Hamiltonian by a two-layer circuit composed of small random unitaries $U_{i,i+1}$ acting on pairs of nearest-neighbor patches~\cite{schuster2024random}.

We consider two instantiations of the random energies $J^i_z$ and the small random unitaries $U_{i,i+1}$, which will allow us to realize unitary $k$-designs and PRUs, respectively.
To realize approximate unitary $k$-designs, we draw each $U_{i,i+1}$ from a strong approximate unitary $k$-design on $2\xi$ qubits~\cite{schuster2025strong}, and each $J^i_z$ from an exact $k$-wise independent function on $\xi$ bits~\cite{fn_discretize,wegman1981new}.
To realize PRUs, we draw each small random unitary $U_{i,i+1}$ from a strong PRU on $2\xi$ qubits with security against any $\poly n$-time quantum adversary~\cite{ma2025construct,schuster2025strong}, and each $J^i_z$ from a pseudorandom function (PRF) on $\xi$ bits~\cite{zhandry2021PRF,fn_discretize}. 
In both cases, we then consider the unitary ensemble $\mathcal{E} = \{ e^{-iHt} \}$ where $H$ is drawn randomly as described and $t = \pi$.

\begin{figure}
    \centering
    \includegraphics[width=0.9\linewidth]{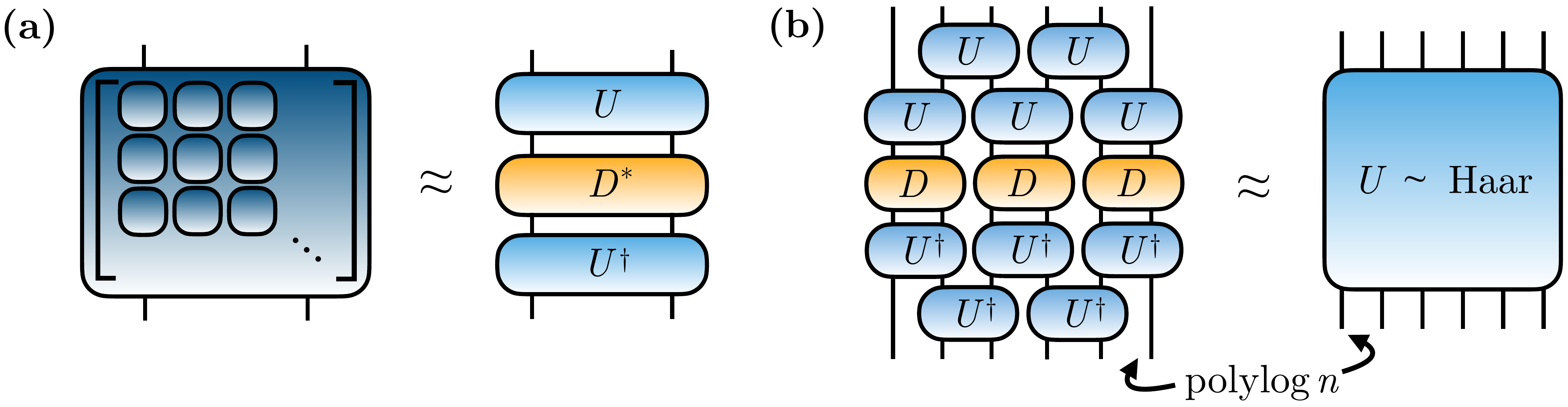}
    \caption{Graphical depiction of the asymptotic decomposition of random matrix ensembles into independent components. \textbf{(a)} Classical random matrix ensembles arising from Wigner matrices with i.i.d. elements converge to a deterministic diagonal distribution, conjugated by Haar-random eigenbasis transformations. \textbf{(b)} Our results show that Haar-random unitaries are indistinguishable from certain ensembles generated by transformations that act on subsystems of size $\mathcal{O}(\poly\log n)$.}
    \label{fig:random matrix}
\end{figure}

Our first main result is that the first time-evolution ensemble above forms an approximate unitary $k$-design with measurable error $\varepsilon$ whenever $\xi = \Omega(\log nk/\varepsilon)$.
\begin{theorem}[Unitary $k$-designs from time-evolution under nearly-local Hamiltonians] \label{thm:designs}
    Consider the random two-layer Hamiltonian ensemble  in which each small random unitary is drawn from a strong $\frac{\varepsilon}{n}$-approximate unitary $k$-design and each small random phase is an exact $k$-wise independent function.
    Then the random time-evolution $U = e^{-iHt}$ for $t = \pi$ forms an $\varepsilon$-approximate unitary $k$-design for any $\xi \geq 2 \log_2(nk/\varepsilon) + \mathcal{O}(1)$.
\end{theorem}
\noindent Our second main result is that the second time-evolution ensemble above forms a PRU with security against any $\poly n$-time quantum adversary whenever $\xi = \omega(\log n)$.
\begin{theorem}[PRUs from time-evolution under nearly-local Hamiltonians] \label{thm:PRUs}
    Consider the random two-layer Hamiltonian ensemble  in which each small random unitary is drawn from a strong PRU ensemble and each small random phase is drawn from a strong PRF ensemble.
    Then the random time-evolution $U = e^{-iHt}$ for $t = \pi$ forms a PRU, for any $\xi = \omega(\log n)$.
\end{theorem}

Our proofs of Theorems~\ref{thm:designs} and~\ref{thm:PRUs} leverage the path-recording framework~\cite{ma2025construct} and make several significant advancement beyond previous work.
In particular, despite the appearance of the same two-layer random unitary circuit in our Hamiltonian as in~\cite{schuster2024random}, the gluing results of~\cite{schuster2024random} strictly do not apply to our setting.
This owes to the presence of the inverse unitary circuit in our Hamiltonian definition, which is necessary to ensure that the Hamiltonian is Hermitian.
Instead, we prove Theorems~\ref{thm:designs} and~\ref{thm:PRUs} by establishing a new random unitary gluing lemma that incorporates \emph{conjugation} by a Haar-random unitary and its inverse.
Namely, for any subsystems $\mathsf a, \mathsf b, \mathsf c, \mathsf d$, we prove that the product of random unitaries, $U_{\mathsf{bc}}^\dagger (U_{\mathsf{ab}} \otimes U_{\mathsf{cd}} ) U_{\mathsf{bc}}$, is indistinguishable from a random unitary on the union of the subsystems, $U_{\mathsf{abcd}}$.
This holds up to a measurable error that decays exponentially in the size of $\mathsf b$ and $\mathsf c$.
Despite the recent proliferation of gluing lemmas in the literature~\cite{grevink2025will,zhang2025designs,cui2025unitary,foxman2025random,schuster2025strong,schuster2025hardness,ananth2025pseudorandom}, this is the first gluing lemma to our knowledge to incorporate the commonplace phenomenon of unitary conjugation.
This incorporation requires a substantially new and more intricate proof approach, owing to the presence of the inverse unitary $U^\dagger_{\mathsf{bc}}$.

To apply this gluing lemma to our ensemble of interest, we decompose each Hamiltonian time-evolution as follows,
\begin{equation} \nonumber
    e^{-iH \pi} = \left( \otimes_{i \in \text{even}} U_{i,i+1} \right)^\dagger  \left( \otimes_{i \in \text{odd}} U_{i,i+1} \right)^\dagger  \left( \otimes_{i} D_i \right)  \left( \otimes_{i \in \text{odd}} U_{i,i+1} \right)  \left( \otimes_{i \in \text{even}} U_{i,i+1} \right),
\end{equation}
where $D_i \equiv e^{i\pi \sum_i \sum_z J^i_z \dyad{z}_i}$ are uniformly random phases in the computational basis on each patch.
We proceed in three steps.
First, following earlier work~\cite{gu2024simulating}, we show that the uniform random phases $D_i$, are indistinguishable from the spectrum of a Haar-random unitary on patch $i$.
This allows us to replace each $D_i$---within its conjugation by the random unitary $U_{i,i+1}$---with a small Haar-random unitary $U_i$~\cite{gu2024simulating}.
%
Second, we apply our conjugation gluing lemma $n/2\xi$ times, to replace each of the resulting $U^\dagger_{i,i+1}(U_i \otimes U_{i+1}) U_{i,i+1}$ with a Haar-random unitary $U'_{i,i+1}$. 
In total, these two steps yield the random unitary ensemble $(\otimes_{i \in \text{even}} U_{i,i+1} )^\dagger (\otimes_{i \in \text{odd}} U'_{i,i+1} ) (\otimes_{i \in \text{even}} U_{i,i+1} )$.
To complete the proof, we apply our gluing lemma $n/2\xi$ additional times, in sequence from left to right, to glue each $U'_{i,i+1}$ into a larger random unitary acting on all qubits to its left.
This results in a Haar-random unitary on the entire system of $n$ qubits, as desired.
We perform a detailed analysis of the errors incurred in this approach in Appendix~\ref{app:adaptive}, which  yields Theorems~\ref{thm:designs} and~\ref{thm:PRUs}. 
%

A natural follow-up question concerns the generality of our results in Theorems~\ref{thm:designs} and~\ref{thm:PRUs}. 
Does the fast formation of random unitaries hold for generic ensembles of $\mathcal{O}(\log n)$-local Hamiltonians, beyond the specific ensemble considered here?
We provide some evidence towards a positive answer in the following theorem, which extends our proof of the formation of unitary designs and PRUs to a much broader range of Hamiltonian spectra and evolution times.
For simplicity, we derive this extension only for additive error unitary designs and PRUs with non-adaptive security, unlike our main results which hold for measurable error unitary designs and adaptive security.
%
\begin{theorem}[Additive error designs from generic spectra and evolution times] \label{thm:generic-designs}
    Consider the random two-layer Hamiltonian ensemble in Theorem~\ref{thm:designs}, but in which each random diagonal term, $\sum_z J^i_z \dyad{z}_i$, is replaced with any fixed Hamiltonian $H_i$.
    The resulting time-evolution $U = e^{-iHt}$ forms an additive-error $\varepsilon$-approximate unitary $k$-design for any $H_i$ and any time $t$ such that $|\! \Tr(e^{-iH_i t})|^2 = o(\varepsilon/nk^2)$ for all $i$.
\end{theorem}
\begin{theorem}[Non-adaptive PRUs from generic spectra and evolution times] \label{thm:generic-PRUs}
    Consider the random two-layer Hamiltonian ensemble in Theorem~\ref{thm:PRUs}, but in which each random diagonal term, $\sum_z J^i_z \dyad{z}_i$, is replaced with any fixed Hamiltonian $H_i$.
    The resulting time-evolution $U = e^{-iHt}$ forms a PRU with non-adaptive security for any $H_i$ and any time $t$ such that $|\! \Tr(e^{-iH_i t})|^2 = o(1/\poly n)$ for all $i$.
\end{theorem}
\noindent Our proofs of both theorems are provided in Appendix~\ref{app:non-adaptive}.

We conclude this section with a few remarks.
First, we emphasize that the Hamiltonians we consider are efficiently described, by only a $\poly n$ number of random variables.
They are also efficient to simulate on a quantum device for any evolution time.
Both of these properties follow immediately from our use of strong unitary $k$-designs~\cite{schuster2025strong}, $k$-wise independent functions~\cite{wegman1981new}, strong PRUs~\cite{ma2025construct,schuster2025strong}, and PRFs~\cite{zhandry2021PRF}, all of which have efficient descriptions and can be efficiently simulated.
This is not guaranteed for general $\omega(\log n)$-local Hamiltonians, which may require $\omega(\poly n)$ parameters to specify.

Second, one might wonder: Is it necessary for the Hamiltonian of the system to be completely random, or might a smaller amount of randomness suffice?
Indeed, a common setting for thermalization considers the scenario where the Hamiltonian is fixed, and the time is randomized instead.
Unfortunately, we find that nearly all of the randomness of a unitary design composed of Hamiltonian time-evolution must come from uncertainty in the Hamiltonian itself:
\begin{proposition}
    [Impossibility of unitary designs from efficient temporal ensembles] \label{prop:temp}
    Consider any unitary ensemble composed of time-evolution by at most $N_H$ many Hamiltonians for evolution times between $0$ and $T$.
    Let $h \equiv \max_H \lVert H \rVert_\infty$ denote the maximum spectral norm of any Hamiltonian in the ensemble.
    Then the maximum evolution time $T$ must be at least $T = \Omega(2^{nk}/(N_H h (k+1)!))$.
\end{proposition}
\noindent The required time is exponential in the number of qubits $n$ for any time-evolution ensemble that involves a subexponential number of Hamiltonians.
The ensemble can only be efficient if there are at least $N_H = \Omega(2^{nk})$ Hamiltonians, in which case the denominator in lower bound on $T$ cancels the numerator. 

Our proof of Proposition~\ref{prop:temp} follows from a simple argument.
Any approximate unitary design must contain at least $\mathcal{O}(nk)$ bits of randomness~\cite{brandao2016local}.
Meanwhile, randomizing the evolution time over a window $[0,T]$ can produce at most $\mathcal{O}(\log_2(\lVert H \rVert_\infty T))$ bits of randomness, since there are at most $\lVert H \rVert_\infty T$ unitaries $e^{-iHt}$ that are distinguishable from all other such unitaries.
For any normalized Hamiltonian and polynomial evolution time, this number of bits is at most $\mathcal{O}(\log n)$, which is exponentially smaller than the required bits of randomness $\mathcal{O}(nk)$.
Hence, the remaining $\mathcal{O}(nk)$ bits of randomness must come from a lack of knowledge of an extensive number of parameters of the Hamiltonian itself.

\section{Discussion}

Our results raise several interesting questions for future work.
First, while our results rule out unitary designs and PRUs from constant-local Hamiltonian dynamics, can state designs and pseudorandom states still emerge from constant-local Hamiltonian dynamics?
Our no-go theorems do not extend to random states, since a random Hamiltonian ensemble could be chosen adversarially to have precisely zero energy with respect to a fixed initial state.
Second, can one define an appropriate analog of unitary designs and PRUs to capture constant-local Hamiltonian time-evolutions?
In particular, the Scrooge ensemble provides a promising approach for time-evolved quantum states in certain scenarios~\cite{jozsa1994lower,mark2024maximum}; however, it is unclear how to extend such notions to unitary time-evolutions.
Our results also suggest that uncertainty in the parameters of the Hamiltonian itself is essential to generating enough randomness to form universal random ensembles, but it is not clear how to incorporate this insight in the context of the Scrooge ensembles.
Answering these questions could yield important practical applications, for example in analog quantum experiments and simulations, and fundamental new insights on the nature of complexity and randomness in physical systems.

\section*{Acknowledgments}

We are grateful to Daniel Mark, John Preskill, and Adam Shaw for insightful discussions. T.S. acknowledges support from the Walter Burke Institute for Theoretical Physics at Caltech. T.S. and H.H. acknowledge support from the U.S. Department of Energy, Office of Science, National Quantum Information Science Research Centers, Quantum Systems Accelerator. The Institute for Quantum Information and Matter is an NSF Physics Frontiers Center.

\clearpage
\bibliographystyle{ieeetr}
\bibliography{refs}

\begin{thebibliography}{10}

\bibitem{wigner1967random}
E.~P. Wigner, ``Random matrices in physics,'' {\em SIAM review}, vol.~9, no.~1, pp.~1--23, 1967.

\bibitem{bohigas1984characterization}
O.~Bohigas, M.-J. Giannoni, and C.~Schmit, ``Characterization of chaotic quantum spectra and universality of level fluctuation laws,'' {\em Physical review letters}, vol.~52, no.~1, p.~1, 1984.

\bibitem{muller2004semiclassical}
S.~M{\"u}ller, S.~Heusler, P.~Braun, F.~Haake, and A.~Altland, ``Semiclassical foundation of universality in quantum chaos,'' {\em Physical review letters}, vol.~93, no.~1, p.~014103, 2004.

\bibitem{deutsch1991quantum}
J.~M. Deutsch, ``Quantum statistical mechanics in a closed system,'' {\em Phys. Rev. A}, vol.~43, p.~2046, 1991.

\bibitem{rigol2008thermalization}
M.~Rigol, V.~Dunjko, and M.~Olshanii, ``Thermalization and its mechanism for generic isolated quantum systems,'' {\em Nature}, vol.~452, no.~7189, pp.~854--858, 2008.

\bibitem{emerson2003pseudo}
J.~Emerson, Y.~S. Weinstein, M.~Saraceno, S.~Lloyd, and D.~G. Cory, ``Pseudo-random unitary operators for quantum information processing,'' {\em science}, vol.~302, no.~5653, pp.~2098--2100, 2003.

\bibitem{gross2007evenly}
D.~Gross, K.~Audenaert, and J.~Eisert, ``Evenly distributed unitaries: On the structure of unitary designs,'' {\em Journal of mathematical physics}, vol.~48, no.~5, 2007.

\bibitem{dankert2005efficient}
C.~Dankert, ``Efficient simulation of random quantum states and operators,'' {\em arXiv preprint quant-ph/0512217}, 2005.

\bibitem{dankert2009exact}
C.~Dankert, R.~Cleve, J.~Emerson, and E.~Livine, ``Exact and approximate unitary 2-designs and their application to fidelity estimation,'' {\em Physical Review A}, vol.~80, no.~1, p.~012304, 2009.

\bibitem{brandao2016local}
F.~G. Brandao, A.~W. Harrow, and M.~Horodecki, ``Local random quantum circuits are approximate polynomial-designs,'' {\em Communications in Mathematical Physics}, vol.~346, pp.~397--434, 2016.

\bibitem{nakata2016efficient}
Y.~Nakata, C.~Hirche, M.~Koashi, and A.~Winter, ``Efficient unitary designs with nearly time-independent hamiltonian dynamics,'' {\em arXiv preprint arXiv:1609.07021}, 2016.

\bibitem{haah2024efficient}
J.~Haah, Y.~Liu, and X.~Tan, ``Efficient approximate unitary designs from random pauli rotations,'' {\em arXiv preprint arXiv:2402.05239}, 2024.

\bibitem{chen2024incompressibility}
C.-F. Chen, J.~Haah, J.~Haferkamp, Y.~Liu, T.~Metger, and X.~Tan, ``Incompressibility and spectral gaps of random circuits,'' {\em arXiv preprint arXiv:2406.07478}, 2024.

\bibitem{guo2024complexity}
S.~Guo, M.~Sasieta, and B.~Swingle, ``Complexity is not enough for randomness,'' {\em SciPost physics}, vol.~17, no.~6, p.~151, 2024.

\bibitem{laracuente2024approximate}
N.~LaRacuente and F.~Leditzky, ``Approximate unitary $ k $-designs from shallow, low-communication circuits,'' {\em arXiv preprint arXiv:2407.07876}, 2024.

\bibitem{schuster2024random}
T.~Schuster, J.~Haferkamp, and H.-Y. Huang, ``Random unitaries in extremely low depth,'' {\em Science}, vol.~389, no.~6755, pp.~92--96, 2025.

\bibitem{lami2025anticoncentration}
G.~Lami, J.~De~Nardis, and X.~Turkeshi, ``Anticoncentration and state design of random tensor networks,'' {\em Physical Review Letters}, vol.~134, no.~1, p.~010401, 2025.

\bibitem{grevink2025will}
L.~Grevink, J.~Haferkamp, M.~Heinrich, J.~Helsen, M.~Hinsche, T.~Schuster, and Z.~Zimbor{\'a}s, ``Will it glue? on short-depth designs beyond the unitary group,'' {\em arXiv preprint arXiv:2506.23925}, 2025.

\bibitem{cui2025unitary}
L.~Cui, T.~Schuster, F.~Brand{\~a}o, and H.-Y. Huang, ``Unitary designs in nearly optimal depth,'' {\em arXiv preprint arXiv:2507.06216}, 2025.

\bibitem{schuster2025strong}
T.~Schuster, F.~Ma, A.~Lombardi, F.~Brand\~{a}o, and H.-Y. Huang, ``Strong random unitaries and fast scrambling,'' {\em arXiv preprint arXiv:2509.26310}, 2025.

\bibitem{foxman2025random}
B.~Foxman, N.~Parham, F.~Vasconcelos, and H.~Yuen, ``Random unitaries in constant (quantum) time,'' {\em arXiv preprint arXiv:2508.11487}, 2025.

\bibitem{schuster2025hardness}
T.~Schuster, D.~Kufel, N.~Y. Yao, and H.-Y. Huang, ``Hardness of recognizing phases of matter,'' {\em Forthcoming}, 2025.

\bibitem{zhang2025designs}
Y.~Zhang, S.~Vijay, Y.~Gu, and Y.~Bao, ``Designs from magic-augmented clifford circuits,'' {\em arXiv preprint arXiv:2507.02828}, 2025.

\bibitem{ji2018pseudorandom}
Z.~Ji, Y.-K. Liu, and F.~Song, ``Pseudorandom quantum states,'' in {\em Advances in Cryptology--CRYPTO 2018: 38th Annual International Cryptology Conference, Santa Barbara, CA, USA, August 19--23, 2018, Proceedings, Part III 38}, pp.~126--152, Springer, 2018.

\bibitem{brakerski2019pseudo}
Z.~Brakerski and O.~Shmueli, ``(pseudo) random quantum states with binary phase,'' in {\em Theory of Cryptography Conference}, pp.~229--250, Springer, 2019.

\bibitem{metger2024simple}
T.~Metger, A.~Poremba, M.~Sinha, and H.~Yuen, ``Simple constructions of linear-depth t-designs and pseudorandom unitaries,'' {\em arXiv preprint arXiv:2404.12647}, 2024.

\bibitem{chen2024efficient}
C.-F. Chen, A.~Bouland, F.~G. Brand{\~a}o, J.~Docter, P.~Hayden, and M.~Xu, ``Efficient unitary designs and pseudorandom unitaries from permutations,'' {\em arXiv preprint arXiv:2404.16751}, 2024.

\bibitem{ma2025construct}
F.~Ma and H.-Y. Huang, ``How to construct random unitaries,'' in {\em Proceedings of the 57th Annual ACM Symposium on Theory of Computing}, pp.~806--809, 2025.

\bibitem{ananth2025pseudorandom}
P.~Ananth, J.~Bostanci, A.~Gulati, and Y.-T. Lin, ``Pseudorandom unitaries in the haar random oracle model,'' in {\em Annual International Cryptology Conference}, pp.~301--333, Springer, 2025.

\bibitem{emerson2005scalable}
J.~Emerson, R.~Alicki, and K.~{\.Z}yczkowski, ``Scalable noise estimation with random unitary operators,'' {\em Journal of Optics B: Quantum and Semiclassical Optics}, vol.~7, no.~10, p.~S347, 2005.

\bibitem{knill2008randomized}
E.~Knill, D.~Leibfried, R.~Reichle, J.~Britton, R.~B. Blakestad, J.~D. Jost, C.~Langer, R.~Ozeri, S.~Seidelin, and D.~J. Wineland, ``Randomized benchmarking of quantum gates,'' {\em Physical Review A}, vol.~77, no.~1, p.~012307, 2008.

\bibitem{elben2023randomized}
A.~Elben, S.~T. Flammia, H.-Y. Huang, R.~Kueng, J.~Preskill, B.~Vermersch, and P.~Zoller, ``The randomized measurement toolbox,'' {\em Nature Reviews Physics}, vol.~5, no.~1, pp.~9--24, 2023.

\bibitem{guta2020fast}
M.~Gu{\c{t}}{\u{a}}, J.~Kahn, R.~Kueng, and J.~A. Tropp, ``Fast state tomography with optimal error bounds,'' {\em Journal of Physics A: Mathematical and Theoretical}, vol.~53, no.~20, p.~204001, 2020.

\bibitem{huang2020predicting}
H.-Y. Huang, R.~Kueng, and J.~Preskill, ``Predicting many properties of a quantum system from very few measurements,'' {\em Nature Physics}, vol.~16, no.~10, pp.~1050--1057, 2020.

\bibitem{zhao2021fermionic}
A.~Zhao, N.~C. Rubin, and A.~Miyake, ``Fermionic partial tomography via classical shadows,'' {\em Physical Review Letters}, vol.~127, no.~11, p.~110504, 2021.

\bibitem{arute2019quantum}
F.~Arute, K.~Arya, R.~Babbush, D.~Bacon, J.~C. Bardin, R.~Barends, R.~Biswas, S.~Boixo, F.~G. Brandao, D.~A. Buell, {\em et~al.}, ``Quantum supremacy using a programmable superconducting processor,'' {\em Nature}, vol.~574, no.~7779, pp.~505--510, 2019.

\bibitem{morvan2023phase}
A.~Morvan, B.~Villalonga, X.~Mi, S.~Mandra, A.~Bengtsson, P.~Klimov, Z.~Chen, S.~Hong, C.~Erickson, I.~Drozdov, {\em et~al.}, ``Phase transition in random circuit sampling,'' {\em arXiv preprint arXiv:2304.11119}, 2023.

\bibitem{abanin2025constructive}
D.~A. Abanin, R.~Acharya, L.~Aghababaie-Beni, G.~Aigeldinger, A.~Ajoy, R.~Alcaraz, I.~Aleiner, T.~I. Andersen, M.~Ansmann, F.~Arute, {\em et~al.}, ``Constructive interference at the edge of quantum ergodic dynamics,'' {\em arXiv preprint arXiv:2506.10191}, 2025.

\bibitem{hayden2007black}
P.~Hayden and J.~Preskill, ``Black holes as mirrors: quantum information in random subsystems,'' {\em JHEP}, vol.~2007, no.~09, p.~120, 2007.

\bibitem{bouland2019computational}
A.~Bouland, B.~Fefferman, and U.~Vazirani, ``Computational pseudorandomness, the wormhole growth paradox, and constraints on the ads/cft duality,'' {\em arXiv preprint arXiv:1910.14646}, 2019.

\bibitem{kim2020ghost}
I.~Kim, E.~Tang, and J.~Preskill, ``The ghost in the radiation: Robust encodings of the black hole interior,'' {\em Journal of High Energy Physics}, vol.~2020, no.~6, pp.~1--65, 2020.

\bibitem{akers2022black}
C.~Akers, N.~Engelhardt, D.~Harlow, G.~Penington, and S.~Vardhan, ``The black hole interior from non-isometric codes and complexity,'' {\em arXiv preprint arXiv:2207.06536}, 2022.

\bibitem{akers2024holographic}
C.~Akers, A.~Bouland, L.~Chen, T.~Kohler, T.~Metger, and U.~Vazirani, ``Holographic pseudoentanglement and the complexity of the ads/cft dictionary,'' {\em arXiv preprint arXiv:2411.04978}, 2024.

\bibitem{yang2025complexity}
L.~Yang and N.~Engelhardt, ``The complexity of learning (pseudo) random dynamics of black holes and other chaotic systems,'' {\em Journal of High Energy Physics}, vol.~2025, no.~3, pp.~1--65, 2025.

\bibitem{aharonov2021quantum}
D.~Aharonov, J.~Cotler, and X.-L. Qi, ``Quantum algorithmic measurement,'' {\em Nature communications}, vol.~13, no.~1, pp.~1--9, 2022.

\bibitem{huang2021quantum}
H.-Y. Huang, M.~Broughton, J.~Cotler, S.~Chen, J.~Li, M.~Mohseni, H.~Neven, R.~Babbush, R.~Kueng, J.~Preskill, {\em et~al.}, ``Quantum advantage in learning from experiments,'' {\em Science}, vol.~376, no.~6598, pp.~1182--1186, 2022.

\bibitem{cotler2023information}
J.~Cotler, T.~Schuster, and M.~Mohseni, ``Information-theoretic hardness of out-of-time-order correlators,'' {\em Physical Review A}, vol.~108, no.~6, p.~062608, 2023.

\bibitem{BF13}
W.~{Brown} and O.~{Fawzi}, ``{Decoupling with random quantum circuits},'' {\em Comm. Math. Phys.}, vol.~340, p.~867, 2015.

\bibitem{yoshida2017efficient}
B.~Yoshida and A.~Kitaev, ``Efficient decoding for the hayden-preskill protocol,'' {\em arXiv preprint arXiv:1710.03363}, 2017.

\bibitem{landsman2019verified}
K.~A. Landsman, C.~Figgatt, T.~Schuster, N.~M. Linke, B.~Yoshida, N.~Y. Yao, and C.~Monroe, ``Verified quantum information scrambling,'' {\em Nature}, vol.~567, no.~7746, pp.~61--65, 2019.

\bibitem{blok2021quantum}
M.~Blok, V.~Ramasesh, T.~Schuster, K.~O'Brien, J.~Kreikebaum, D.~Dahlen, A.~Morvan, B.~Yoshida, N.~Yao, and I.~Siddiqi, ``Quantum information scrambling on a superconducting qutrit processor,'' {\em Physical Review X}, vol.~11, no.~2, p.~021010, 2021.

\bibitem{BF12}
W.~Brown and O.~Fawzi, ``{Scrambling speed of random quantum circuits},'' 2012.

\bibitem{nahum2017entgrowth}
A.~Nahum, J.~Ruhman, S.~Vijay, and J.~Haah, ``{Quantum Entanglement Growth Under Random Unitary Dynamics},'' {\em Phys. Rev. X}, vol.~7, p.~031016, 2017.

\bibitem{nahum2018operator}
A.~Nahum, S.~Vijay, and J.~Haah, ``Operator spreading in random unitary circuits,'' {\em Physical Review X}, vol.~8, no.~2, p.~021014, 2018.

\bibitem{schuster2022many}
T.~Schuster, B.~Kobrin, P.~Gao, I.~Cong, E.~T. Khabiboulline, N.~M. Linke, M.~D. Lukin, C.~Monroe, B.~Yoshida, and N.~Y. Yao, ``Many-body quantum teleportation via operator spreading in the traversable wormhole protocol,'' {\em Physical Review X}, vol.~12, no.~3, p.~031013, 2022.

\bibitem{boixo2018characterizing}
S.~Boixo, S.~V. Isakov, V.~N. Smelyanskiy, R.~Babbush, N.~Ding, Z.~Jiang, M.~J. Bremner, J.~M. Martinis, and H.~Neven, ``Characterizing quantum supremacy in near-term devices,'' {\em Nature Physics}, vol.~14, no.~6, pp.~595--600, 2018.

\bibitem{bouland2019complexity}
A.~Bouland, B.~Fefferman, C.~Nirkhe, and U.~Vazirani, ``On the complexity and verification of quantum random circuit sampling,'' {\em Nature Physics}, vol.~15, no.~2, p.~159, 2019.

\bibitem{brown2018second}
A.~R. Brown and L.~Susskind, ``Second law of quantum complexity,'' {\em Physical Review D}, vol.~97, no.~8, p.~086015, 2018.

\bibitem{haferkamp2022linear}
J.~Haferkamp, P.~Faist, N.~B. Kothakonda, J.~Eisert, and N.~Yunger~Halpern, ``Linear growth of quantum circuit complexity,'' {\em Nature Physics}, vol.~18, no.~5, pp.~528--532, 2022.

\bibitem{maldacena2016bound}
J.~Maldacena, S.~H. Shenker, and D.~Stanford, ``A bound on chaos,'' {\em Journal of High Energy Physics}, vol.~2016, no.~8, pp.~1--17, 2016.

\bibitem{roberts2017chaos}
D.~A. Roberts and B.~Yoshida, ``Chaos and complexity by design,'' {\em Journal of High Energy Physics}, vol.~2017, no.~4, p.~121, 2017.

\bibitem{cotler2017chaos}
J.~Cotler, N.~Hunter-Jones, J.~Liu, and B.~Yoshida, ``Chaos, complexity, and random matrices,'' {\em Journal of High Energy Physics}, vol.~2017, no.~11, pp.~1--60, 2017.

\bibitem{pilatowsky2024hilbert}
S.~Pilatowsky-Cameo, I.~Marvian, S.~Choi, and W.~W. Ho, ``Hilbert-space ergodicity in driven quantum systems: Obstructions and designs,'' {\em Physical Review X}, vol.~14, no.~4, p.~041059, 2024.

\bibitem{gu2024simulating}
A.~Gu, Y.~Quek, S.~Yelin, J.~Eisert, and L.~Leone, ``Simulating quantum chaos without chaos,'' {\em arXiv preprint arXiv:2410.18196}, 2024.

\bibitem{kong2021charge}
L.~Kong and Z.-W. Liu, ``Charge-conserving unitaries typically generate optimal covariant quantum error-correcting codes,'' {\em arXiv preprint arXiv:2102.11835}, 2021.

\bibitem{kong2022near}
L.~Kong and Z.-W. Liu, ``Near-optimal covariant quantum error-correcting codes from random unitaries with symmetries,'' {\em PRX Quantum}, vol.~3, no.~2, p.~020314, 2022.

\bibitem{li2024efficient}
Z.~Li, H.~Zheng, and Z.-W. Liu, ``Efficient quantum pseudorandomness under conservation laws,'' {\em arXiv preprint arXiv:2411.04893}, 2024.

\bibitem{li2024designs}
Z.~Li, H.~Zheng, J.~Liu, L.~Jiang, and Z.-W. Liu, ``Designs from local random quantum circuits with su (d) symmetry,'' {\em PRX Quantum}, vol.~5, no.~4, p.~040349, 2024.

\bibitem{liu2024unitary}
H.~Liu, A.~Hulse, and I.~Marvian, ``Unitary designs from random symmetric quantum circuits,'' {\em arXiv preprint arXiv:2408.14463}, 2024.

\bibitem{mitsuhashi2025unitary}
Y.~Mitsuhashi, R.~Suzuki, T.~Soejima, and N.~Yoshioka, ``Unitary designs of symmetric local random circuits,'' {\em Physical Review Letters}, vol.~134, no.~18, p.~180404, 2025.

\bibitem{haah2025short}
J.~Haah, ``Short remarks on shallow unitary circuits,'' {\em arXiv preprint arXiv:2504.14005}, 2025.

\bibitem{pappalardi2022eigenstate}
S.~Pappalardi, L.~Foini, and J.~Kurchan, ``Eigenstate thermalization hypothesis and free probability,'' {\em Physical Review Letters}, vol.~129, no.~17, p.~170603, 2022.

\bibitem{fava2025designs}
M.~Fava, J.~Kurchan, and S.~Pappalardi, ``Designs via free probability,'' {\em Physical Review X}, vol.~15, no.~1, p.~011031, 2025.

\bibitem{cotler2023emergent}
J.~S. Cotler, D.~K. Mark, H.-Y. Huang, F.~Hern{\'a}ndez, J.~Choi, A.~L. Shaw, M.~Endres, and S.~Choi, ``Emergent quantum state designs from individual many-body wave functions,'' {\em PRX quantum}, vol.~4, no.~1, p.~010311, 2023.

\bibitem{mark2024maximum}
D.~K. Mark, F.~Surace, A.~Elben, A.~L. Shaw, J.~Choi, G.~Refael, M.~Endres, and S.~Choi, ``Maximum entropy principle in deep thermalization and in hilbert-space ergodicity,'' {\em Physical Review X}, vol.~14, no.~4, p.~041051, 2024.

\bibitem{anschuetz2025strongly}
E.~R. Anschuetz, C.-F. Chen, B.~T. Kiani, and R.~King, ``Strongly interacting fermions are nontrivial yet nonglassy,'' {\em Physical Review Letters}, vol.~135, no.~3, p.~030602, 2025.

\bibitem{wiebe2014hamiltonian}
N.~Wiebe, C.~Granade, C.~Ferrie, and D.~G. Cory, ``Hamiltonian learning and certification using quantum resources,'' {\em Physical Review Letters}, vol.~112, no.~19, p.~190501, 2014.

\bibitem{evans2019scalable}
T.~J. Evans, R.~Harper, and S.~T. Flammia, ``Scalable bayesian {H}amiltonian learning,'' {\em arXiv preprint arXiv:1912.07636}, 2019.

\bibitem{haah2021optimal}
J.~Haah, R.~Kothari, and E.~Tang, ``Optimal learning of quantum {Hamiltonians} from high-temperature {Gibbs} states,'' {\em arXiv preprint arXiv:2108.04842}, 2021.

\bibitem{huang2023learningb}
H.-Y. Huang, Y.~Tong, D.~Fang, and Y.~Su, ``Learning many-body hamiltonians with heisenberg-limited scaling,'' {\em Physical Review Letters}, vol.~130, no.~20, p.~200403, 2023.

\bibitem{dutkiewicz2023advantage}
A.~Dutkiewicz, T.~E. O'Brien, and T.~Schuster, ``The advantage of quantum control in many-body {Hamiltonian} learning,'' {\em arXiv preprint arXiv:2304.07172}, 2023.

\bibitem{bakshi2024structure}
A.~Bakshi, A.~Liu, A.~Moitra, and E.~Tang, ``Structure learning of hamiltonians from real-time evolution,'' in {\em 2024 IEEE 65th Annual Symposium on Foundations of Computer Science (FOCS)}, pp.~1037--1050, IEEE, 2024.

\bibitem{mezard2009information}
M.~Mezard and A.~Montanari, {\em Information, physics, and computation}.
\newblock Oxford University Press, 2009.

\bibitem{steinacker2010emergent}
H.~Steinacker, ``Emergent geometry and gravity from matrix models: an introduction,'' {\em Classical and Quantum Gravity}, vol.~27, no.~13, p.~133001, 2010.

\bibitem{maldacena2016remarks}
J.~Maldacena and D.~Stanford, ``{Remarks on the Sachdev-Ye-Kitaev model},'' {\em Physical Review D}, vol.~94, no.~10, p.~106002, 2016.

\bibitem{qi2019quantum}
X.-L. Qi and A.~Streicher, ``Quantum epidemiology: operator growth, thermal effects, and syk,'' {\em Journal of High Energy Physics}, vol.~2019, no.~8, 2019.

\bibitem{lin2022infinite}
H.~Lin and L.~Susskind, ``Infinite temperature's not so hot,'' {\em arXiv preprint arXiv:2206.01083}, 2022.

\bibitem{berkooz2025cordial}
M.~Berkooz and O.~Mamroud, ``A cordial introduction to double scaled syk,'' {\em Reports on Progress in Physics}, vol.~88, no.~3, p.~036001, 2025.

\bibitem{swingle2024bosonic}
B.~Swingle and M.~Winer, ``Bosonic model of quantum holography,'' {\em Physical Review B}, vol.~109, no.~9, p.~094206, 2024.

\bibitem{chen2024sparse}
C.-F. Chen, A.~M. Dalzell, M.~Berta, F.~G. Brand{\~a}o, and J.~A. Tropp, ``Sparse random hamiltonians are quantumly easy,'' {\em Physical Review X}, vol.~14, no.~1, p.~011014, 2024.

\bibitem{fn_discretize}
To be precise, for unitary $k$-designs, we specify each $J^i_z$ up to $m = \Omega(\log k/\varepsilon)$ bits and draw each bit from a $k$-wise independent function $f: \{0,1\}^\xi \rightarrow \{0,1\}$. For PRUs, we specify each $J^i_z$ up to $m = \omega(\log n)$ bits and draw each bit from a pseudorandom function $f: \{0,1\}^\xi \rightarrow \{0,1\}$.

\bibitem{wegman1981new}
M.~N. Wegman and J.~L. Carter, ``New hash functions and their use in authentication and set equality,'' {\em Journal of computer and system sciences}, vol.~22, no.~3, pp.~265--279, 1981.

\bibitem{zhandry2021PRF}
M.~Zhandry, ``How to construct quantum random functions,'' {\em J. ACM}, vol.~68, aug 2021.

\bibitem{jozsa1994lower}
R.~Jozsa, D.~Robb, and W.~K. Wootters, ``Lower bound for accessible information in quantum mechanics,'' {\em Physical Review A}, vol.~49, no.~2, p.~668, 1994.

\bibitem{fulton2013representation}
W.~Fulton and J.~Harris, {\em Representation theory: a first course}, vol.~129.
\newblock Springer Science \& Business Media, 2013.

\bibitem{goodman2009symmetry}
R.~Goodman, N.~R. Wallach, {\em et~al.}, {\em Symmetry, representations, and invariants}, vol.~255.
\newblock Springer, 2009.

\bibitem{collins2006integration}
B.~Collins and P.~{\'S}niady, ``Integration with respect to the haar measure on unitary, orthogonal and symplectic group,'' {\em Communications in Mathematical Physics}, vol.~264, no.~3, pp.~773--795, 2006.

\bibitem{zhandry2019record}
M.~Zhandry, ``How to record quantum queries, and applications to quantum indifferentiability,'' in {\em Annual International Cryptology Conference}, pp.~239--268, Springer, 2019.

\bibitem{elben2022randomized}
A.~Elben, S.~T. Flammia, H.-Y. Huang, R.~Kueng, J.~Preskill, B.~Vermersch, and P.~Zoller, ``The randomized measurement toolbox,'' {\em Nature Review Physics}, 2022.

\end{thebibliography}

\resumetoc



\let\oldaddcontentsline\addcontentsline
\renewcommand{\addcontentsline}[3]{}
\let\addcontentsline\oldaddcontentsline

\onecolumngrid
\newpage
\appendix

\noindent 
\textbf{\LARGE{}Appendices}
\vspace{2em}


\tableofcontents

\section{Review of Haar-random unitaries, unitary designs, and pseudorandom unitaries}

In this section, we provide a more detailed review of the definitions and key properties of Haar-random unitaries, unitary designs, and pseudorandom unitaries (PRUs).

\subsection{Haar-random unitaries}

We now review properties of Haar-random unitaries, which form a uniformly and maximally random distribution against which we compare other ensembles. In this context, the Haar measure is given by the unique translation-invariant measure on the unitary group $U(N)$.
\begin{definition}[Moments of the unitary group]
    Given a linear operator $X$ acting on $nk$ qubits, the $k$-th moment with respect to $U(2^n)$ is defined via the twirl over the unitary group:
    \begin{equation}
    \Phi_{H}(X) = \int dU \, U^{\otimes k} X (U^\dagger)^{\otimes k},
    \end{equation}
    where we have left out the implicit dependence on $k$.
\end{definition}
This definition captures the expectation over how $X$ transforms when we apply $k$ copies of a Haar-random unitary $U$. By linearity, this is equivalent to considering an arbitrary entangled input state $\ket{\psi}$ on $nk + m$ qubits. The structure of these moments can be determined using representation theoretic properties of the unitary group. By definition, these expectation values must be invariant under any unitary change of basis applied to each individual copy. Applying Schur-Weyl duality \cite{fulton2013representation, goodman2009symmetry} yields that the moments must be of the following form:
\begin{fact}[Explicit form in terms of permutations]
    For any linear operator $X$ acting on $nk$ qubits, the $k$-th moment with respect to the unitary group can be written in the form 
    \begin{equation}
    \Phi_{H}(X) = \sum_{\sigma, \tau \in S_k} c_{\sigma, \tau} \Tr(\sigma X) \cdot \tau,
    \end{equation}
    where the coefficients $c_{\sigma,\tau}$ depend on $k$ and the Hilbert space dimension $2^n$.
\end{fact}
In principle, the coefficients $c_{\sigma, \tau}$ can be computed exactly using combinatorial formulas derived from the Weingarten calculus of the unitary group \cite{collins2006integration}. However, the main property of the moments that we use is their relationship to the symmetric group on $k$ copies of the physical systems. In particular, the techniques in our work rely on the key observation that a large fraction of the moments are supported on the \emph{distinct subspace}, which is spanned by the symmetrized vectors corresponding to elements of $[N]^k_\mathrm{dist}$:


%
\begin{definition}[Distinct subspace~\cite{metger2024simple}] \label{def:dist subspace}
    The projector onto the \emph{distinct subspace} is given by the operator
    \vspace{-2mm}\begin{equation}
        \Pi^{\mathrm{dist}} = \sum_{\mathbf{x} \in [N]^k_\mathrm{dist}} \ketbra*{x},
    \end{equation}
    where $[N]^k_\mathrm{dist} = \{\mathbf{x} = (x^{(1)}, x^{(2)}, ..., x^{(k)}): x^{(i)} \neq x^{(j)} \text{ for } i \neq j\}.$ The subspace has dimension $\mathfrak{D} = (2^n)!/(2^n-k)!$ which obeys $1-k^2/2^n \leq \mathfrak{D}/2^{nk} \leq 1$.
\end{definition}
Notably, nearly all bitstrings are distinct when $k^2 \ll 2^n$. Thus for any input state on the physical system, ``most'' of the moments are also supported on this subspace. This fact was originally applied to analyze state designs in \cite{brakerski2019pseudo}, and since then has been extended to various frameworks for analyzing unitary ensembles \cite{metger2024simple,ma2025construct,cui2025unitary}. 

In addition, is often useful to apply these restrictions to a \emph{locally} distinct subspace corresponding to some subsystems of subextensive size, similar to the techniques of \cite{cui2025unitary}. In particular, for subsystems of size $\xi$, the corresponding local distinct subspace is given as the space spanned by the symmetrized vectors corresponding to elements of $[2^\xi]^k_\mathrm{dist}$ on each subsystem:
\begin{definition}[Local distinct subspace \cite{cui2025unitary}] \label{def:restrict local distinct}
    For any system of $n$ qubits divided into patches of $\xi$ qubits each,
    let $\text{\emph{loc-dist}} = \{ \mathbf{x} : x_a^{(i)} \neq x_a^{(j)} \text{{ for all }} a \text{{ and }} i \neq j \}$ denote the set of locally distinct $\mathbf{x}$.
    The \emph{local distinct subspace} corresponding to $S$ is defined by the projector
    \vspace{-2mm}
    \begin{equation}
        \Pi^{\mathrm{dist}}_{\mathrm{loc}} =  \sum_{x \in \text{\emph{loc-dist}}} \dyad{x}.  \vspace{-1mm}       
    \end{equation}
    The local distinct subspace has dimension $\mathfrak{D}_{\text{\emph{loc}}} = ((2^\xi)!/(2^\xi-k)!)^{n/\xi}$ which obeys $1-nk^2/2^\xi \xi \leq \mathfrak{D}_{\text{\emph{loc}}} \leq 1$.
\end{definition}
In the following sections, we discuss how the properties of the $k$-wise twirls over different ensembles relate to notions of indistinguishability from Haar-random.

\subsection{Unitary designs}

\subsubsection{Additive error and parallel indistinguishability}

The most well-studied notion of statistical distance for unitary ensembles is additive error, which is given by the difference between the moments in diamond norm: 
%
%
\begin{definition}[Unitary $k$-design up to additive error]
    An ensemble of unitaries $\mathcal{E}$ is a \emph{unitary $k$-design} if it reproduces the first $k$ moments of the Haar measure. In addition, given some $\varepsilon > 0$, the ensemble $\mathcal{E}$ is an approximate unitary $k$-design up to \emph{additive error} $\varepsilon$ if
    \begin{equation}
    \norm{\Phi_\mathcal{E} - \Phi_H}_\diamond \leq \varepsilon,
    \end{equation}
    where we have used the abbreviated notation
    \begin{equation}
    \Phi_\mathcal{E}(X) = \E_{U \sim \mathcal{E}} U^{\otimes k} X (U^\dagger)^{\otimes k}
    \end{equation}
    to denote the $k$-th moment over the unitary ensemble $\mathcal E$.
\end{definition}
\noindent Recall that the diamond norm is defined via $\lVert \Phi - \Phi' \rVert_\diamond = \max_\rho \lVert \Phi(\rho) - \Phi'(\rho) \rVert_1$, where the maximization is over all input states $\rho$ on $nk+m$ qubits, including an arbitrarily large ancillary system of size $m$. This norm thus measures the maximum distinguishability of $\Phi_\mathcal{E}$ and $\Phi_H$, which apply $k$ copies of a random unitary $U$ in parallel.
It follows that the additive error definition is equivalent to the maximum distinguishability of $\mathcal E$ from the Haar ensemble in any quantum experiment that queries $k$ copies of the unitary $U$ in parallel.
%
%
\begin{fact}[Additive error is equivalent to parallel indistinguishability]
    An ensemble $\mathcal{E}$ is an approximate unitary $k$-design up to additive error $\varepsilon$ if and only if for any quantum algorithm making a single query to $U^{\otimes k}$, i.e.~$k$ parallel queries to $U$, the output states when $U$ is sampled from $\mathcal{E}$ versus the Haar ensemble are $\varepsilon$-close in trace distance.
\end{fact}

\subsubsection{Measurable error and adaptive indistinguishability}

Certain settings in quantum complexity theory and cryptography demand indistinguishability under more general quantum experiments that can query $k$ copies of $U$. It is known that this condition is strictly stronger than indistinguishability under queries to $U^{\otimes k}$, i.e., $k$ parallel queries to $U$ \cite{metger2024simple, ma2025construct}. 

In particular, we consider quantum experiments making $k$ queries to $U$ which can apply the unitary in sequence and with arbitrary quantum operations in between each query. This is much more powerful as it enables the quantum experiment to learn some properties about $U$, then adaptively probe the unitary $U$ based on the properties it has learned. The formulation of a general quantum experiment that makes $k$ queries to $U$ is given as follows.
\begin{definition}[Quantum experiments with $k$ queries to $U$] \label{def: qu expt}
    A quantum experiment with $k$ queries to a unitary consists of:
    \begin{enumerate}
        \item An initial state preparation $\ket{\psi_0}$ on registers $A \otimes B$, where $A$ has dimension $2^n$ and $B$ is an auxiliary register of dimension $2^m$. Without loss of generality, we can set $\ket{\psi_0} = \ket{0^n} \otimes \ket{0^m}$.
        \item For $i = 1, \ldots, k$: 
        \begin{itemize}
            \item Apply a unitary $T_i$ to registers $A \otimes B$.
            \item Apply the unknown unitary $U$ to register $A$.
        \end{itemize}
        \item Apply a final unitary $T_{k+1}$ and measure to obtain classical outcome.
    \end{enumerate}
    %
\end{definition}
To capture indistinguishability under the most powerful quantum experiments that make up to $k$ queries to a unitary $U$, we require a corresponding stronger notion of approximation error, the measurable error \cite{cui2025unitary}. This quantity is given by the maximum distinguishability between a random unitary ensemble and the Haar ensemble over all possible 
$k$-query quantum experiments. Moreover, whenever this error is large, there exists a quantum experiment that can distinguish the two ensembles. 
\begin{definition}[Unitary $k$-design up to measurable error]
    Let $\varepsilon > 0$. An ensemble of unitaries $\mathcal{E}$ is an approximate unitary $k$-design up to measurable error $\varepsilon$ if for any quantum experiment with $k$ queries to $U$, the output states when $U$ is sampled from $\mathcal{E}$ versus the Haar ensemble are $\varepsilon$-close in trace distance, i.e.~
    \begin{equation}
        \sup_{T_1 \cdots T_{k+1}} \norm{\rho_\mathcal{E} - \rho_H}_1 \leq \varepsilon,
    \end{equation}
    where we have used the notation 
    \begin{equation}
        \rho_\mathcal{E} = \expect_{U \sim \mathcal{E}} \left[ T_{k+1} [U \otimes \mathbbm{1}_m] T_k \cdots T_2 [U \otimes \mathbbm{1}_m] T_1 |0^{n+m} \rangle\! \langle 0^{n+m} | T_1^\dagger [U^\dagger \otimes \mathbbm{1}_m] T_2^\dagger \cdots T_k^\dagger [U^\dagger \otimes \mathbbm{1}_m] T_{k+1}^\dagger 
        \right]
    \end{equation}
    to denote the expected output state of a general quantum experiment that queries $U$ $k$ times.
\end{definition}
%
In addition, it is shown in \cite{cui2025unitary} that these experiments are equivalent to ones which use only $k$ parallel queries, followed by a postselection on a Bell state on $nk$ qubits. We refer to Section IV.B of the Supplementary Material \cite{cui2025unitary} for more details.

\subsection{Pseudorandom unitaries}

In the context of quantum cryptography, pseudorandom unitary (PRU) ensembles are families of unitaries which can be efficiently constructed but are \emph{computationally} indistinguishable from Haar-random \cite{ji2018pseudorandom} by any adversary with only polynomial resources:
\begin{definition}[Pseudorandom unitary ensemble]
    A family of unitaries defined via $\mathcal U_n = \{U_\alpha \in U(2^n)\}_{\alpha}$ is pseudorandom with \emph{non-adaptive} security if it satisfies the following:
    \begin{enumerate}
        \item (Efficient implementation) There exists an efficient algorithm $Q$ such that for any $U_\alpha \in \mathcal{U}_n$ and $\ket{\psi} \in \mathbb{C}^{2^n}$, $Q(\alpha, \ket{\psi}) = U_\alpha\ket{\psi}$.
        \item (Computational indistinguishability) For $U_\alpha \sim \mathcal{U}_n$ and $k \in \mathrm{poly}(n)$, $U_\alpha$ is computationally indistinguishable from $k$ copies of a Haar-random unitary operator. More precisely, for any efficient quantum algorithm $\calT$ which uses at most $k$ copies of $U_\alpha$,
        \begin{equation}
        \abs{\expect\limits_{\alpha}[\calT[U_\alpha] (|0\rangle^k) = 1] - \mathop{\expect}\limits_{U \sim \Haar} [\calT[U] (|0\rangle^k) = 1]} = \mathrm{negl}(n).
        \label{condition:pr}
        \end{equation}
    \end{enumerate}
\end{definition}
Here $\mathrm{negl}(n)$ denotes an inverse superpolynomial scaling $1/\omega(\poly n)$. While the efficient implementation condition is straightforward to check given a description of $\mathcal{U}_n$ and a protocol to implement $U_\alpha$, it is more difficult to directly bound the computational distinguishability of $\mathcal{U}_n$ from the Haar ensemble. However, it is sufficient to show that it is indistinguishable from a Haar-random unitary transformation even with access to arbitrary generalized measurements:
\begin{fact}[Sufficient condition for computational indistinguishability]
    Suppose for any $k = O(\poly n)$, $\mathcal{U}_n$ forms an approximate unitary $k$-design up to $\epsilon = 1/\omega(\poly n)$. Then $U_\alpha \sim \mathcal{U}_n$ is also \emph{computationally} indistinguishable from a Haar-random unitary.
\end{fact}
It is common to consider settings which allow either at most $k$ parallel or $k$ sequential queries, corresponding to closeness in additive or measurable error, respectively \cite{metger2024simple}. The existence of PRUs in the adaptive setting remained an open question until very recently, when it was resolved in \cite{ma2025construct} assuming the existence of quantum-secure pseudorandom functions on bitstrings. In the following section, we provide an overview of techniques introduced in \cite{ma2025construct} in order to analyze queries to, or applications of, random unitary transformations.


\section{Review of the path-recording framework}

In this section we summarize the path-recording framework, which was introduced in \cite{ma2025construct} as a framework for analyzing the security of PRUs, including the usual adaptive setting as well as a stronger version where the adversary has access to both $U$ and $U^\dagger$.

\subsection{Relation state registers}

In order to analyze the ensembles of interest, it is often useful to consider additional working registers which can be used to approximately reconstruct the twirling operation. In particular, we introduce \emph{relation state} registers to record how the basis states are shuffled by the Haar twirl. Here we define a relation $R$ as a multiset of ordered pairs $\{(x_1, y_1), \dots, (x_k, y_k)\},$ where $(x_i, y_i) \in [N]^2$. The \emph{size} $|R|$ equals the number of pairs counting multiplicities.
\begin{definition}[Sets of relations]
Let $\mathcal{R}$ denote the set of all relations and $\mathcal{R}_k$ the set of all length-$k$ relations. Then define
\begin{align}
    \Dom(R) &= \{x \in [N]: \exists y \text{ such that } (x,y) \in R\},\\
    \Im(R) &= \{y \in [N]: \exists x \text{ such that } (x,y) \in R\},\\
    \Dom_J(R) &= \{x_J \in [2^{\abs{J}}]: \exists x, y \text{ such that } (x, y) \in R \text{ and } \forall j \in J, x^{(j)} = x_J^{(j)}\},\\
    \Im_J(R) &= \{y_J \in [2^{\abs{J}}]: \exists x, y \text{ such that } (x, y) \in R \text{ and } \forall j \in J, y^{(j)} = y_J^{(j)}\},
\end{align}
where $J \subseteq [n]$ corresponds to a subset of the physical registers.
\end{definition}
Each relation $R$ corresponds to a \emph{relation state} in the symmetric subspace.
\begin{definition}[Relation states] 
\label{def:relation-states}
For relation $R = \{(x_1,y_1),\dots,(x_k,y_k)\}$, define
\begin{align}
    \ket{R} = \frac{\sum_{\pi \in S_k} \ket{x_{\pi(1)},y_{\pi(1)},\dots,x_{\pi(k)},y_{\pi(k)}}}{\sqrt{k! \cdot \prod_{(x,y) \in [N]^2} \num(R,(x,y))!}},
\end{align}
where $\num(R,(x,y))$ denotes the multiplicity of pair $(x,y)$ in $R$.
\end{definition}
The relation states form an orthonormal basis for the symmetric subspace of $(\mathbb{C}^{N^2})^{\otimes k}$. When all pairs in $R$ are distinct, the normalization simplifies to $1/\sqrt{k!}$.
In addition, we often consider relation state registers which may be of \emph{variable} length:
\begin{definition}[Variable-length registers]
    For $k \geq 0$, let $\mathsf{R}^{(k)}$ be a register with Hilbert space $\mathcal{H}_{\mathsf{R}^{(k)}} = (\mathbb{C}^N \otimes \mathbb{C}^N)^{\otimes k}$. Define the variable-length register $\mathsf{R}$ with infinite-dimensional Hilbert space
    \begin{align}
        \mathcal{H}_{\mathsf{R}} = \bigoplus_{t=0}^\infty \mathcal{H}_{\mathsf{R}^{(k)}} = \bigoplus_{t=0}^\infty (\mathbb{C}^N \otimes \mathbb{C}^N)^{\otimes k}.
    \end{align}
\end{definition}
Moreover, we can decompose $\mathsf{R}^{(k)} = (\mathsf{R}^{(k)}_{\mathsf{X}}, \mathsf{R}^{(k)}_{\mathsf{Y}})$ where $\mathsf{R}^{(k)}_{\mathsf{X}} = \ket{x_1,\dots,x_k}$ and $\mathsf{R}^{(k)}_{\mathsf{Y}} = \ket{y_1,\dots,y_k}$. Due to the direct sum structure, relation states of different lengths are orthogonal.

\subsection{Restrictions on relation states}

We now discuss restricted sets of relation states. These play an analogous role in our analysis of the path-recording oracle to that of the distinct subspace projection described in Definition~\ref{def:dist subspace} for the analysis of the Haar twirl in previous work \cite{brakerski2019pseudo,metger2024simple,cui2025unitary}.
\begin{definition}[Restricted relation sets]
We define variants of restricted relation sets as follows.
\begin{itemize}
    \item $\mathcal{R}_k^{\text{inj}}$: injective relations where $(y_1,\dots,y_k) \in [N]^k_{\text{dist}}$
    \item $\mathcal{R}_k^{\text{bij}}$: bijective relations where $(x_1,\dots,x_k), (y_1,\dots,y_k) \in [N]^k_{\text{dist}}$
\end{itemize}
We also define $\mathcal{R}^{\text{inj}} = \bigcup_{t=0}^N \mathcal{R}_k^{\text{inj}}$ and $\mathcal{R}^{\text{bij}} = \bigcup_{t=0}^N \mathcal{R}_k^{\text{bij}}$.
\end{definition}
We now introduce projection operators which act on the space of relation states.
\begin{definition}[Projectors and extensions]
Define the projector onto relation states of length $k$:
\begin{align}
    \Pi^{\mathcal{R}}_k = \sum_{R \in \mathcal{R}_k} \ketbra{R} = \Pi^{N^2,k}_{\ssym},
\end{align}
where $\Pi^{N^2,k}_{\ssym}$ projects onto the symmetric subspace of $(\mathbb{C}^{N^2})^{\otimes k}$. 
%
\end{definition}

%

We introduce additional notation for constructions involving two variable-length registers such as $\mathsf{L}$ and $\mathsf{R}$ which correspond to forward and inverse applications of the same unitary.

\begin{definition}[Length projectors] \label{notation:pi-leq-t}
For integers $\ell,r \geq 0$, let $\Pi_{\ell,r}$ project onto $\mathcal{H}_{\mathsf{L}^{(\ell)}} \otimes \mathcal{H}_{\mathsf{R}^{(r)}}$. For integer $t \geq 0$, let $\Pi_{\leq k}$ project onto $\bigoplus_{\ell,r \geq 0: \ell + r \leq k} \mathcal{H}_{\mathsf{L}^{(\ell)}} \otimes \mathcal{H}_{\mathsf{R}^{(r)}}$.
\end{definition}

\begin{definition}[Length-restricted operators]
For operator $B$ acting on registers $\mathsf{L}$ and $\mathsf{R}$, define
\begin{align}
    B_{\ell,r} &= B \cdot \Pi_{\ell,r},\\
    B_{\leq k} &= B \cdot \Pi_{\leq k}.
\end{align}
We adopt the convention that $B_{\leq k}^\dagger = (B_{\leq k})^\dagger$.
\end{definition}
\begin{definition}[Restricted order pairs of relation sets]
    Let $\calR^{2,\dist}$ be the set of all ordered pairs of relations $(L,R) \in \calR^2$ where $L \cup R = \{(x_1, y_1), \ldots, (x_t, y_t))\}$ is a bijective relation, i.e., $x_1, \ldots, x_t$ are distinct and $y_1, \ldots, y_t$ are distinct.
\end{definition}
\begin{definition}[Bijective-relation projectors] \label{def:bij-proj}
    Define the projectors
    \begin{align}
        \Pi^{\bij}_{\gsL \gsR} \coloneq \sum_{(L,R) \in \mathcal{R}^{2,\dist}} \ketbra*{L}_{\gsL} \otimes \ketbra*{R}_{\gsR}, \quad\quad \Pi^{\bij}_{\leq k, \gsL \gsR} \coloneq \Pi^{\bij}_{\gsL \gsR} \cdot \Pi_{\leq k, \gsL \gsR} = \Pi_{\leq k, \gsL \gsR} \cdot \Pi^{\bij}_{\gsL \gsR},
    \end{align}
    where the projector $\Pi_{\leq k, \gsL \gsR}$ is the maximum-length projector given by \cref{notation:pi-leq-t}.
\end{definition}
%
In addition, we state the following properties, which are often useful for bounding statistical distances between quantum states.
\begin{fact} \label{lemma:trace to lp 2}
    For any pure states $\ket{u}, \ket{v}$ with  $\langle u | u \rangle, \langle v | v \rangle \leq 1$,
    \begin{equation} \label{eq: state to op bound}
        \big\lVert \dyad{u} - \dyad{v} \big\rVert_1 \leq 2 \big\lVert \ket{u} - \ket{v} \big\rVert_2,
    \end{equation}
    where the $\norm{\cdot}_2$ on the RHS denotes the vector $2$-norm.
\end{fact}
\begin{lemma}[Gentle measurement lemma] \label{lemma:gentle measurement}
    \begin{equation} \label{eq: gentle measurement lemma}
    \big\lVert \Pi \rho \Pi - \rho \big\rVert_1 \leq 2 \sqrt{ 1 - \Tr( \Pi \rho ) }.
\end{equation}
\end{lemma}
\begin{lemma}[Sequential gentle measurement (Lemma 2.3 of \cite{ma2025construct})]  \label{lem:seq-gentleM-pure}
    Let $\ket{\psi}$ be a normalized state, $P_1,\dots,P_k$ be projectors, and $U_1,\dots,U_k$ be unitaries.
    \begin{align}
        \norm{U_k \ldots U_1 \ket{\psi} -  P_k U_k \ldots P_{1} U_1 \ket{\psi}}_2 \leq k \sqrt{1 - \norm{P_k U_k \ldots P_{1} U_1 \ket{\psi}}_2^2}.
    \end{align}
\end{lemma}

\subsection{The path-recording oracle}

We now present definitions and useful technical results for the path-recording oracle framework introduced in \cite{ma2025construct}. This construction builds upon a line of work \cite{zhandry2019record} from the cryptography literature on using purification oracles to analyze notions of computational security against quantum adversaries. 
\begin{definition}[Oracle adversaries]
\label{def:oracle-adversary}
A $k$-query oracle adversary $\calT$ is parameterized by a sequence of $(n+m)$-qubit unitaries $(T_1,\dots,T_{k+1})$ acting on registers $(\sA,\sB)$, where $\mathsf{A}$ is the $n$-qubit query register and $\sB$ is an $m$-qubit ancilla, and a sequence of oracle queries $U_1, \ldots, U_k$ where each $U_i \in \{U, U^\dagger\}$. The state after $k$ queries is
\begin{align}
    \ket{\calT_k^{U}}_{\gsA\gsB} = T_{k+1} [U_k \otimes \mathbbm{1}_m] \cdots T_2 [U_1 \otimes \mathbbm{1}_m] T_1 \ket{0^{n+m}}_{\gsA\gsB}.
\end{align}
\end{definition}
The path-recording oracle $V$ proposed in \cite{ma2025construct} efficiently simulates a Haar-random unitary $U$ under both queries to $U$ and $U^\dagger$. Informally, $V$ can be understood as a unitary oracle which constructs an approximate purification for the Haar twirl, such that tracing out over the auxiliary working registers results in a state which indistinguishable from the output of the Haar twirl. To define the path-recording oracle $V$, we separately construct the left and right parts, or $V^L$ and $V^R$, respectively.
\begin{definition}[Left and right parts of $V$] \label{def:V-sym-PRO}
    Let $V^L$ be the linear operator that acts as follows. For $x \in [N]$ and $(L,R) \in \mathcal{R}^{2,\leq N-1}$,
    \begin{equation}
        V^L \cdot \ket{x}_{\gsA} \ket{L}_{\gsL} \ket{R}_{\gsR} = \sum_{\substack{y \in [N]:\\ y\not\in \Im(L \cup R)}} \frac{1}{\sqrt{N - \abs{\Im(L \cup R)}}} \ket{y}_{\gsA} \ket{L \cup \{(x,y)\}}_{\gsL} \ket{R}_{\gsR}.
    \end{equation}
    Define $V^R$ to be the linear operator such that for all $y \in [N]$ and $(L,R) \in \mathcal{R}^{2,\leq N-1}$,
    \begin{equation}
        V^R \cdot \ket{y}_{\gsA} \ket{L}_{\gsL} \ket{R}_{\gsR} = \sum_{\substack{x \in [N]:\\ x\not\in \Dom(L \cup R)}} \frac{1}{\sqrt{N - \abs{\Dom(L \cup R)}}} \ket{x}_{\gsA} \ket{L}_{\gsL} \ket{R \cup \{(x, y)\} }_{\gsR}.
    \end{equation}
    By construction, $V^L$ and $V^R$ take states in $\Id_{\gsA} \otimes \Pi^{\calR^2}_{\leq i, \gsL \gsR}$ to $\Id_{\gsA} \otimes \Pi^{\calR^2}_{\leq i+1, \gsL \gsR}$.
\end{definition}
\begin{fact}[Claim 14 \cite{ma2025construct}]
$V^L$ and $V^R$ are partial isometries.
\end{fact}
\begin{definition}[Path-recording oracle $V$]
\label{def:symmetric-V}
    The path-recording oracle is the operator $V$ defined as
    \begin{align}
        V &= V^L \cdot (\Id - V^R \cdot V^{R,\dagger}) + (\Id - V^L \cdot V^{L,\dagger}) \cdot V^{R,\dagger}.
    \end{align}
    By construction, $V$ and $V^\dagger$ take states in $\Id_{\gsA} \otimes \Pi^{\calR^2}_{\leq i, \gsL \gsR}$ to $\Id_{\gsA} \otimes \Pi^{\calR^2}_{\leq i+1, \gsL \gsR}$ for any integer $i \geq 0$.
\end{definition}
\begin{fact}[Claim 15 in \cite{ma2025construct}]
$V$ is a partial isometry.
\end{fact}
%

\subsection{The partial path-recording oracle}
Another useful construction from \cite{ma2025construct} is the partial path-recording oracle $W$, which is a restricted version of the full path-recording oracle $V$. The operator $W$ only acts nontrivially on a subspace and maps the orthogonal subspace to zero. The subspace is defined based on $\calR^{2,\dist}$.
Similar to $V$, the partial path-recording oracle $W$ contains a left part $W^L$ and a right part $W^R$.
\begin{definition}[$W^L$ and $W^R$]
\label{def:ternary-W-action}
    Define $W^L$ to be the linear map such that for any $(L,R) \in \calR^{2,\dist}$ and $x \in [N]$ such that $x\not\in \Dom(L \cup R)$,
    \begin{align}
        W^L \cdot \ket{x}_{\gsA} \ket{L}_{\gsL} \ket{R}_{\gsR} = \frac{1}{\sqrt{N-\abs{L \cup R}}} \sum_{\substack{y \in [N]:\\ y\not\in \Im(L \cup R)}} \ket{y}_{\gsA} \ket{L \cup \{(x,y)\}}_{\gsL} \ket{R}_{\gsR}. \label{eq:WL-def}
    \end{align}
    Similarly, define $W^R$ be the linear map such that for any $(L,R) \in \calR^{2,\dist}$ and $y \in [N]$ such that $y\not\in \Im(L \cup R)$,
    \begin{align}
        W^R \cdot \ket{y}_{\gsA} \ket{L}_{\gsL} \ket{R}_{\gsR} = \frac{1}{\sqrt{N-\abs{L \cup R}}} \sum_{\substack{x \in [N]:\\ x\not\in \Dom(L \cup R)}} \ket{x}_{\gsA} \ket{L}_{\gsL} \ket{R \cup \{(x,y)\}}_{\gsR}. \label{eq:WR-def}
    \end{align}
\end{definition}
\begin{definition} \label{def:symmetric-W}
    The partial path-recording oracle is the operator $W$ defined as
    \begin{equation}
        W = W^L + W^{R,\dagger}.
    \end{equation}
\end{definition}
\begin{lemma}[Fact 5 in \cite{ma2025construct}] \label{fact:WLWR-space-leqi}
    For any integer $i \geq 0$, $W^L, W^R$ map states in the subspace associated to the projector $\Id_{\gsA} \otimes \Pi^{\bij}_{\leq i, \gsL \gsR}$ into the subspace associated with the projector $\Id_{\gsA} \otimes \Pi^{\bij}_{\leq i+1, \gsL \gsR}$.
\end{lemma}
\begin{fact}[Claims 9 and 11 in \cite{ma2025construct}]
    $W^L, W^R, W$ are partial isometries.
\end{fact}
It is useful to introduce projectors of the following form, which we use restrict to valid relation states corresponding to the action of $V$ or $W$:
\begin{definition}
    For a partial isometry $G$, let $\Dom(G)$ and $\Im(G)$ denote its domain and image. Let $\Pi^{\Dom(G)} = G^\dagger \cdot G$ and $\Pi^{\Im(G)} = G \cdot G^\dagger$ denote the orthogonal projectors onto $\Dom(G)$ and $\Im(G)$.
\end{definition}
\begin{lemma}[Fact 8 in \cite{ma2025construct}]\label{claim:W-partial-isometry}
    The domain and image of the partial isometry $W$ are given by
    \begin{align}
        \Pi^{\Dom(W)} &= \Pi^{\Dom(W^L)} + \Pi^{\Im(W^R)}, \label{eq:expand-DW}\\
        \Pi^{\Im(W)} &= \Pi^{\Dom(W^R)} + \Pi^{\Im(W^L)}. \label{eq:expand-IW}
    \end{align}
\end{lemma}
%
%
\begin{lemma}[$W$ is a restriction of $V$ (Claim 17 of \cite{ma2025construct})]
\label{claim:relate-W-and-V}
We have
    \begin{align}
        W &= V \cdot \Pi^{\Dom(W)}, \label{eq:relate-W-and-V-goal1}\\
        W^\dagger &= V^\dagger \cdot \Pi^{\Im(W)}. \label{eq:relate-W-and-V-goal2}
    \end{align}
\end{lemma}
%
%
We now present a key property of $W$ which enables us to analyze approximate purifications on a distinct subspace, after twirling by exact unitary $2$-designs.%
\begin{lemma}[Twirling by unitary $2$-design (Lemma 9.2 in \cite{ma2025construct})] \label{lem:twirling-strongPRU}
For any unitary $2$-design $\mathfrak{C}$, and any integer $0 \leq t \leq N-1$, we have
\begin{align}
    \norm{  \E_{C,D \sim \mathfrak{C}} (C_{\gsA} \otimes Q[C,D]_{\gsL \gsR})^\dagger \cdot \Big( \Pi^{\bij}_{\leq t, \gsL \gsR} - \Pi^{\Dom(W)}_{\leq t, \gsA \gsL \gsR}\Big) \cdot (C_{\gsA} \otimes Q[C,D]_{\gsL \gsR}) }_{\infty} &\leq 6t \sqrt{\frac{t}{N}},\\
    \norm{  \E_{C,D \sim \mathfrak{C}} (D^\dagger_{\gsA} \otimes Q[C,D]_{\gsL \gsR})^\dagger \cdot \Big( \Pi^{\bij}_{\leq t, \gsL \gsR} - \Pi^{\Im(W)}_{\leq t, \gsA \gsL \gsR}\Big) \cdot (D^\dagger_{\gsA} \otimes Q[C,D]_{\gsL \gsR}) }_{\infty} &\leq 6t \sqrt{\frac{t}{N}}.
\end{align}
\end{lemma}


\subsection{Approximate purification of the Haar twirl}

We proceed to describe the action of the path-recording oracle on the system.
\begin{definition}[Global state after queries to $V, V^\dagger$]
    For a $k$-query oracle adversary $\calT$ that can perform queries to $U, U^\dagger$, where $b_i \in \{0, 1\}$ correspond to $U, U^\dagger$, respectively, and any $0 \leq i \leq k$, let
    \begin{align}
        \ket{\calT_i^{V}}_{\gsA \gsB \gsL \gsR} = \prod_{i = 1}^k &\Bigg( \Big( (1-b_i) \cdot V_{\gsA \gsL \gsR} + b_i \cdot V_{\gsA \gsL \gsR}^\dagger \Big) \cdot A_{i,\gsA \gsB} \Bigg)  \ket{0^{n+m}}_{\gsA \gsB} \otimes \ket{\varnothing}_{\gsL} \ket{\varnothing}_{\gsR}
    \end{align}
    denote the global state on registers $\sA,\sB,\sL,\sR$ after $\calT$ makes $i$ queries to $V$.
\end{definition}

We will also consider the global purified state after queries to $W, W^\dagger$, where we twirl the input and the output states by two independent random unitaries sampled from a unitary $2$-design. For this purpose, it is convenient to define the purification of two random unitaries $C, D$ which modify the basis of the input and output to the (partial) path-recording oracle.

\begin{definition}
\label{def:init-D-state}
    For any distribution $\mathfrak{C}$ over $n$-qubit unitaries, define the state
    \begin{align}
    \ket{\init(\mathfrak{C})}_{\gsC \gsD} \coloneq \int_{C,D} \sqrt{ d\mu_{\mathfrak{C}}(C) d\mu_{\mathfrak{C}}(D)} \ket{C}_{\gsC} \otimes \ket{D}_{\gsD},
\end{align}
where $\mu_{\mathfrak{C}}(C)$ is the probability measure according to which $C$ is sampled from $\mathfrak{C}$.
\end{definition}

\begin{definition}[Controlled $C,D$ and $Q$]
\label{def:controlled-CDQ}
    Define the following operators
    \begin{align}
        &\mathsf{cC} \coloneq \int_{C} C_{\gsA} \otimes \ketbra*{C}_{\gsC}, \quad \mathsf{cD} \coloneq \int_{D} D_{\gsA} \otimes \ketbra*{D}_{\gsD},\\
    &\mathsf{cQ} \coloneq \int_{C, D} Q[C, D]_{\gsL {\color{gray} ,} \gsR} \otimes \ketbra*{C}_{\gsC} \otimes \ketbra*{D}_{\gsD}.
\end{align}
\end{definition} 

\begin{definition}[Global state after queries to twirled $W, W^\dagger$]
\label{def:twirled-W-state}
    For a $k$-query adversary $\calT$ that can perform queries to $U, U^\dagger$, where $b_i \in \{0, 1\}$ correspond to $U, U^\dagger$, let
    \begin{align}
        \ket{\calT_0^{W,\mathfrak{C}}} &= \ket{0^n}_{\gsA} \ket{0^m}_{\gsB} \ket{\varnothing}_{\gsL} \ket{\varnothing}_{\gsR} \ket{\mathsf{init}(\mathfrak{C})}_{\gsC \gsD}.
    \end{align}
    For $i$ from $1$ to $t$, let
    \begin{align}
        \ket{\calT_i^{W,\mathfrak{C}}} = &\Big( (1-b_i) \cdot (\scD \cdot W \cdot \scC) + b_i \cdot (\scD \cdot W \cdot \scC)^\dagger \Big) \cdot A_i \cdot \ket{\calT^{W,\mathfrak{C}}_{i-1}}.
    \end{align}
\end{definition}
Finally, we present key properties of the path-recording oracle $V$ which enable them to be used to simulate the action of Haar-random unitaries.
\begin{lemma}[$W$ is indistinguishable from $V$ after twirling (Lemma 9.3 in \cite{ma2025construct})] \label{lem:closeness-AWD-and-PhiVt}
Let $\mathfrak{C}$ be any exact unitary $2$-design. For any $k$-query oracle adversary $\calT$ that can query $\mathcal{O}, \mathcal{O}^\dagger$,
    \begin{align}
        \norm{\Tr_{- \sA \sB} \ketbra*{\calT^{W, \mathfrak{C}}_k}_{\gsA \gsB \gsL \gsR \gsC \gsD}, \Tr_{- \sA \sB} \ketbra*{\calT^{V}_k}_{\gsA \gsB \gsL \gsR}}_1 \leq \frac{18k}{N^{1/8}}. \label{eq:TD-phi-psi}
    \end{align}
\end{lemma}
\begin{lemma}[Two-sided unitary invariance (Claim 16 in \cite{ma2025construct})]
\label{claim:two-sided-invariance}
    For any integer $0 \leq k \leq N-1$ and any pair of $n$-qubit unitaries $C,D$,
    \begin{align}
        \norm{D_{\gsA} \cdot V_{\leq k} \cdot C_{\gsA} \otimes Q[C,D]_{\gsL \gsR}
        - Q[C,D]_{\gsL \gsR} \cdot V_{\leq k}}_{\infty} & \leq 16\sqrt{\frac{2k(k+1)}{N}},\\
        \norm{C_{\gsA}^\dagger \cdot (V^\dagger)_{\leq k} \cdot D_{\gsA}^\dagger \otimes Q[C,D]_{\gsL \gsR}
        - Q[C,D]_{\gsL \gsR} \cdot (V^\dagger)_{\leq k}}_{\infty} & \leq 16\sqrt{\frac{2k(k+1)}{N}}.
    \end{align}
\end{lemma}
%
%
\begin{theorem}[$V$ is indistinguishable from a Haar-random unitary (Theorem 8 in \cite{ma2025construct})]\label{theorem:haar-cho-strong}
     For any $k$-query oracle adversary $\calT$ that can can query $\mathcal{O}, \mathcal{O}^\dagger$,
     \begin{align}
        \norm{\E_{\calO \sim \mu_{\mathsf{Haar}}} \ketbra*{\calT_k^{\calO}}_{\gsA \gsB}, \,\,\, \Tr_{\sL \sR}\left( \ketbra*{\calT^{V}_k}_{\gsA \gsB \gsL \gsR} \right)}_1 \leq \frac{18k(k+1)}{N^{1/8}}.
    \end{align}
\end{theorem}
%

\section{Impossibility of random unitaries from constant-local Hamiltonian dynamics} \label{app:nogo}

In this section, we provide the proof of Theorems~\ref{thm:nogodesign} and~\ref{thm:nogoPRU}, which show that ensembles generated by time-evolution under any constant-local Hamiltonians for any lengths of time cannot form approximate unitary 2-designs nor PRUs. 

%
%


\begin{proof}[Proof of Theorems~\ref{thm:nogodesign} and~\ref{thm:nogoPRU}]
    Let $q$ denote the maximum locality of a term in any Hamiltonian $H$ in our ensemble, i.e.~$q = \max_i w[P_i]$ when we decompose $H$ as a sum of Pauli operators, $H = \sum_i h_i P_i$. We describe a protocol which distinguishes a random constant-local Hamiltonian time evolution from a Haar-random unitary transformation as follows.
    
    First, we select a random stabilizer product state $\ket{u}$ and a random Pauli $P_i$ in $H$ with weight less than $q$.
    We then consider the mean square expectation value of $P_i$ when a random product state $\ket{u}$ is time-evolved under $e^{-iHt}$,
    \begin{equation}
        E_{\mathcal{E}^{(q)}} \equiv \E_{U \sim \mathcal{E}^{(q)}} \E_{i}\E_{u} \bra{u} e^{iHt} P_i e^{-iHt} \ket{u}^2
    \end{equation}
    This is equivalent to computing the expectation value of $P_i \otimes P_i$ on two copies of $e^{-iHt} \ket{u}$.
    In a unitary 2-design with additive error $\varepsilon$, we must have $E_{\mathcal{E}} \leq 2 \varepsilon + \mathcal{O}(1/2^n)$.
    In a pseudorandom unitary with security against any $\poly (n)$-time adversary, we must have $E_{\mathcal{E}} \leq 1/\omega(\poly n)$.

    We lower bound $E_\mathcal{E}$ above these values as follows.
    We first apply a standard formula for the twirl over two copies of a random stabilizer product state~\cite{elben2022randomized}, which yields
    \begin{equation}
        \E_u \bra{u} e^{iHt} P_i e^{-iHt} \ket{u}^2 = \frac{1}{2^n} \Tr( e^{iHt} P_i e^{-iHt} \cdot \calW( e^{iHt} P_i e^{-iHt} ) ),
    \end{equation}
    where $\calW(Q) = (1/3)^{w[Q]} Q$ is a quantum channel that multiplies each Pauli operator $Q$ by a factor that decays exponentially in its weight $w[Q]$.
    Intuitively, this factor corresponds to the probability that the Pauli operator commutes with a random stabilizer product measurement basis.
    We therefore obtain the lower bound
    \begin{equation}
        \frac{1}{2^n} \Tr( e^{iHt} P_i e^{-iHt} \cdot \calW( e^{iHt} P_i e^{-iHt} ) ) \geq (1/3)^q \cdot \frac{1}{2^n} \Tr( e^{iHt} P_i e^{-iHt} \cdot \mathcal{P}_{\leq q}( e^{iHt} P_i e^{-iHt} ) ) ,
    \end{equation}
    where $\mathcal{P}_{\leq q}(Q) = \delta_{w[Q] \leq q} Q$ denotes the superoperator that projects onto Pauli strings of weight less than $q$.
    In addition, we further have the lower bound
    \begin{equation}
        \frac{1}{2^n} \Tr( e^{iHt} P_i e^{-iHt} \cdot \mathcal{P}_{\leq q}( e^{iHt} P_i e^{-iHt} ) )  \geq \frac{\Tr( e^{iHt} P_i e^{-iHt}  H )^2}{\Tr(H^2)},
    \end{equation}
    which follows from applying the superoperator $\mathcal{P}_H(Q) = \Tr(QH) H/\text{tr}(H^2)$, which projects a Pauli string onto its overlap with the Hamiltonian $H$, onto the first copy of $e^{iHt} P_i e^{-iHt}$, and then using the fact that $\Tr(H \mathcal{P}_{\leq q}(\cdot)) = \Tr(H (\cdot))$ to simplify the resulting trace since each term of $H$ has weight at most $q$.
    We now observe that the numerator is time-independent, so that we can simplify it via
    \begin{equation}
        \frac{\Tr( e^{iHt} P_i e^{-iHt}  H )^2}{\Tr(H^2)}  = \frac{\Tr( P_i  H )^2}{\Tr(H^2)}  = \frac{h_i^2}{\sum_j h_j^2}.
    \end{equation}
    Taking the expectation over valid Pauli operators $P_i$ with weight less than or equal to $q$ replaces $h_i^2 \rightarrow \frac{1}{M_q} \sum_{j:w[P_j]\leq q} h_j^2$, where $M_q$ is the number of such Paulis.
    In total, this yields
    \begin{equation}
        \E_{U \sim \mathcal{E}} \E_i \E_{u} \bra{u} e^{iHt} P_i e^{-iHt} \ket{u}^2 \geq (1/3)^{q} \cdot \frac{1}{M_q}.
    \end{equation}
    We have $M_q \leq 4^q n$ in one-dimensional systems and $M_q \leq (4n)^q$ in all-to-all connected systems. 
    Hence, we have $\varepsilon \geq \mathcal{O}(1/(12^qn))$ in one-dimensional systems and $\varepsilon \geq \mathcal{O}(1/(12^q n^q))$ for all-to-all connected systems, which completes our proof for unitary designs.
    To achieve $\varepsilon = 1/\omega(\poly n)$, we require $q = \omega(\log n)$ in one-dimensional systems and $q = \omega(1)$ in all-to-all connected systems, which completes our proof  for pseudorandom unitaries.
\end{proof}

\section{Formation of random unitaries from nearly-local Hamiltonian dynamics}

In this section, we present the proofs of our main results on the formation of unitary designs and PRUs in nearly-local Hamiltonians.
We begin by discussing the non-adaptive setting and presenting the proofs of Theorems~\ref{thm:generic-designs} and~\ref{thm:generic-PRUs} in the main text.
We then turn to the adaptive setting and present the proof of Theorems~\ref{thm:designs} and~\ref{thm:PRUs}.


\subsection{Proof of Theorems~\ref{thm:generic-designs} and~\ref{thm:generic-PRUs}: Indistinguishability in non-adaptive quantum experiments} \label{app:non-adaptive}

Our analysis of the non-adaptive setting is organized as follows.
We first prove that in any non-adaptive quantum experiment making $k$ queries, a unitary $U^\dagger D U$ where $U$ is Haar-random and $D$ is any unitary operation, is indistinguishable from a Haar-random unitary up to measurable error $\mathcal{O}(k^2 |\!\Tr(D)|^2)$.
This allows us to replace $U_{i,i+1}^\dagger (H_i \otimes H_{i+1}) U_{i,i+1}$ for each odd $i$ with a Haar-random unitary on $i,i+1$ in Theorems~\ref{thm:generic-designs} and~\ref{thm:generic-PRUs}.
We then prove a gluing lemma for conjugated random unitaries with non-adaptive security.
Our gluing lemma is strictly weaker than the gluing lemma for adaptive security that we prove in the following section (Appendix~\ref{app:adaptive}), but has the advantage of possessing a more explicit proof.
Finally, we introduce an alternative random Hamiltonian time-evolution ensemble for which the proof of non-adaptive security is especially succinct.



\subsubsection{Moments of the random Hamiltonian dynamics}

We define a model with a Haar-random eigenbasis which is invariant under any unitary transformation, and represents the dynamics on a connected subsystem. We will then show that in the limit of large system size, the polynomial-order moments of these dynamics are close to those of Haar-random transformations. 
\begin{definition}[Randomized Hamiltonian dynamics] \label{def:randomdynamics}
    Consider a family of Hamiltonians on systems of $n$ qubits which is given by
    \begin{equation}
    \mathcal{O}_{\mathsf{inv}} = \{H = U \Lambda U^\dagger \,|\, U \sim \Haar, \, \Lambda \sim \calS\},
    \end{equation}
    where $\mathcal{S}$ is the $N$-variate spectral distribution of the eigenvalues of $\mathcal{O}_{\mathsf{inv}}$. We then define unitary ensembles given by \emph{randomized Hamiltonian dynamics} as those generated by evolving under a Hamiltonian $H$ sampled uniformly from $\mathcal{O}$ for a time distributed according to $t \sim P_H$:
    \begin{equation}
       \mathcal{E}_{\mathsf{inv}} = \left\{\left. e^{-iHt} \, \right\vert \, H \sim \mathcal{O}_{\mathsf{inv}}, t \sim P_H \right\}.
    \end{equation}
    In addition, we denote the distribution of the diagonalization of $\mathcal{E}_{\mathsf{inv}}$ by
    \begin{equation}
        \calD = \left\{\left. e^{-i\Lambda t} \, \right\vert \, \Lambda \sim \mathcal{S}, t \sim P_H \right\}.
    \end{equation}
\end{definition}
Since the eigenbasis and spectral distribution of the ensemble are independent of one another, we can average over each of them separately. Our strategy for proving the statistical closeness of the dynamics of this ensemble to the Haar moments is to first consider the average over the eigenbasis transformation using the path-recording framework of \cite{ma2025construct}. 

It is convenient to introduce a modified version of the path-recording oracle which provides an approximate purification in experiments that only query $k$ copies of a unitary in parallel. This setup is equivalent to an adaptive experiment with input state $\ket{\psi}$ on registers $(\otimes_i^k \gsA_i) \otimes \gsB$, in which each step consists of applying $U$ to $\gsA_i$, then applying a postprocessing step $T$ which is simply the permutation operator which shifts the $\gsA_i$ registers by one.
\begin{definition}[Parallel path-recording oracle] \label{def:parallel oracle}
    The $k$-query \emph{parallel} path recording oracle $V^{(k)}$ is defined via its left and right components
    \begin{equation}
    \begin{aligned}
        V^{(k)}_L \cdot \left[\bigotimes_{i=1}^k \ket{x_i}\right] \ket{L}_{\gsL} \ket{R}_{\gsR} &= \sum_{\substack{\mathbf{y} \in [N]^k_\mathrm{dist}:\\ y_i\not\in \Im(L \cup R)}} \frac{1}{Z_L} \ket{y}_{\gsA} \ket{L \cup \{(x,y_i)\}_i}_{\gsL} \ket{R}_{\gsR}, \\
        V^{(k)}_R \cdot \left[\bigotimes_{i=1}^k \ket{y_i} \right] \ket{L}_{\gsL} \ket{R}_{\gsR} &= \sum_{\substack{\mathbf{x} \in [N]^k_\mathrm{dist}:\\ x_i\not\in \Dom(L \cup R)}} \frac{1}{Z_R} \ket{x}_{\gsA} \ket{L}_{\gsL} \ket{R \cup \{(x_i, y)\}_i }_{\gsR},
    \end{aligned}
    \end{equation}
    where the normalization factors $Z_L, Z_R$ are given by 
    \begin{equation}
    \begin{aligned}
        Z_L &= \frac{1}{\sqrt{(N - \abs{\Im(L \cup R)}) (N - \abs{\Im(L \cup R)} - 1) \cdots (N - \abs{\Im(L \cup R)} - k + 1)}}, \\
        Z_R &= \frac{1}{\sqrt{(N - \abs{\Dom(L \cup R)}) (N - \abs{\Dom(L \cup R)} - 1) \cdots (N - \abs{\Dom(L \cup R)} - k + 1)}}.
    \end{aligned}   
    \end{equation}
\end{definition}
\begin{definition}[Parallel partial path-recording oracle]
    Similar to the adaptive case, the $k$-query \emph{parallel} partial path-recording oracle $W^{(k)}$ is defined via an additional restriction on the domain of the left and right components:
    \begin{equation}
        \begin{aligned}
            W^{(k)}_L &= V^{(k)}_L \cdot \prod_{i=1}^k \Pi^{\Dom(W_{L_i})}, \\
            W^{(k)}_R &= V^{(k)}_R \cdot \prod_{i=1}^k \Pi^{\Dom(W_{R_i})},
        \end{aligned}
    \end{equation}
    where $W_{R_i}$ denotes the partial path-recording oracle applied to the $i$-th copy of the system, and the projectors act on the shared path-recording register as well as the $i$-th copy of the physical system.
\end{definition}
%
%
In the remainder of this section, we use ``path-recording oracle'' to refer to the parallel version, unless otherwise specified. In addition, we will use the following property to simplify our analysis:
%
%
\begin{lemma}[$\widetilde{W}^{(k)}$ is indistinguishable from a Haar-random unitary] \label{lemma:twirled w to haar}
    Let $\mathfrak{C}$ be any exact unitary $2$-design. For any $k$-query oracle adversary $\calT$ that can query $\mathcal{O}, \mathcal{O}^\dagger$,
     \begin{align}
        \norm{\E_{\calO \sim {\mathsf{Haar}}} \ketbra*{\calT_k^{\calO}}_{\gsA \gsB} - \Tr_{\gsL \gsR \gsC \gsD} \ketbra*{\calT^{W, \mathfrak{C}}_k}_{\gsA \gsB \gsL \gsR \gsC \gsD}}_1 \leq \frac{36k(k+1)}{N^{1/8}}.
    \end{align}
\end{lemma}
\proof This follows from Lemma~\ref{lem:closeness-AWD-and-PhiVt} and Theorem~\ref{theorem:haar-cho-strong} via an application of the triangle inequality. \qed
\\

We can therefore substitute all instances of the random eigenbasis transformation with applications of $W^{(k)}$, sandwiched by $C, C'$ which are drawn from an exact $2$-design, to obtain an approximate purification of the moment. 
%
We proceed to the main technical result of this section, which is bounding the statistical distance between this purification and the output of $\widetilde{V}^{(k)}$, i.e. the action of $V^{(k)}$ after conjugating by some $C$ drawn from a $2$-design.
\begin{proposition}[Statistics of random eigenbasis dynamics are close to $\widetilde{V}^{(k)}$] \label{prop:random H moments}
    Suppose $D$ is sampled from an arbitrary distribution $\calD$. For any $k$ and input state $\ket{\psi}$ on $nk + m$ qubits, the trace distance between the twirl over the random eigenbasis Hamiltonian dynamics and the action of $\widetilde{V}^{(k)}$ is bounded by
    \begin{equation}
        \norm{\expect_{U \sim \calE_{\mathsf{inv}}} \left[\left(U\right)^{\otimes k} |\psi \rangle \langle \psi | \left(U^\dagger\right)^{\otimes k} \right] - \Tr_{\gsL} |\calT^{\widetilde{V}}_k \rangle \langle \calT^{\widetilde{V}}_k |}_1 \leq \frac{144k(k+1)}{N^{1/8}} + \frac{12k^2}{N} + \expect_{D \sim \calD} \left[ \frac{8k^2}{2^{2n}} \abs{\Tr D}^2 \right],
    \end{equation}
    up to subleading corrections $O(k^2/N^2).$
\end{proposition}
\proof 
We first apply Lemma~\ref{lemma:twirled w to haar} to explicitly compute an approximate purification of the average over the random eigenbasis transformation. For any instance of $D$, we obtain a state of the form
\begin{equation} \label{eq:purify haar avg eigenbasis}
    \ket{\calT_k^{\mathsf{inv}}} = \frac{1}{Z} \expect_C \left[ C^{\otimes k} W^{(k)} \sum_{\mathbf{x,} \mathbf{y} \in [N]^k_{\text{dist}}} \ket{\varnothing}_\gsL \ket{\{(x_i, y_i)\}_{i \in [k]}}_\gsR \left[\bigotimes_{j=1}^k \widetilde{D} \ket{x_j} \bra{y_j}\right] (C^\dagger)^{\otimes k} \right]\ket{\psi},
\end{equation}
where $\widetilde{D}$ is the twirl of $C'D(C')^\dagger$ over $C'$, $Z$ is the normalization factor $(N!/(N-k)!)^{-1/2}$, and we have expanded the action of the inverse direction $(W^{(k)})^\dagger$ acting in the inverse direction. Since this includes a total of $2k$ copies of $U$ or $U^\dagger$, we incur an error which is upper bounded by
\begin{equation}
    \norm{\Tr_{\gsL\gsR} |\calT^{\mathsf{inv}}_k \rangle \langle \calT^{\mathsf{inv}}_k | - \expect_{U \sim \Haar} \left[\left(U  D U^\dagger\right)^{\otimes k} |\psi \rangle \langle \psi | \left(U D^\dagger U^\dagger\right)^{\otimes k} \right]}_1 \leq \frac{144k(k+1)}{2^{n/8}}. \label{eq:random H moment purify bound}
\end{equation}
%
Rather than constructing a purification and change of basis corresponding to the twirl over $C, C',$ we directly impose a restriction on the support of $\ket{\calT^{\mathsf{inv}}_k}$. In particular, we insert the projector
\begin{equation}
    \Pi^{\not\in \mathrm{Im}(W^R)} = \mathbbm{1} - \Pi^{\mathrm{Im}(W^R)}
\end{equation}
before the action of $W^{(k)}$ to obtain
\begin{equation}
    \ket{\calT^{\widetilde{\mathsf{inv}}}_k} = \expect_C \left[ C^{\otimes k} W^{(k)} \Pi^{\not\in \mathrm{Im}(W^R)} \widetilde{D} (W^\dagger)^{(k)} (C^\dagger)^{\otimes k} \right]\ket{\psi}.
\end{equation}
The error from applying this restriction is given by 
\begin{equation}
    \begin{aligned}
        \norm{\Tr_{\gsL\gsR} |\calT^{\mathsf{inv}}_k \rangle \langle \calT^{\mathsf{inv}}_k | - \Tr_{\gsL\gsR} |\calT^{\widetilde{\mathsf{inv}}}_k \rangle \langle \calT^{\widetilde{\mathsf{inv}}}_k |}_1 
        &\leq 2 \norm{\expect_C \left[ C^{\otimes k} W^{(k)} \Pi^{\mathrm{Im}(W^R)} \widetilde{D} (W^\dagger)^{(k)}  (C^\dagger)^{\otimes k} \right]\ket{\psi}}_2,
    \end{aligned} \label{eq:restrict random H moment Im W}
\end{equation}
where we have used the property that the trace distance is monotonic under quantum operations, in addition to Fact~\ref{lemma:trace to lp 2}. The expression on the RHS is upper bounded by summing over all overlap terms in the computational basis from the $k$ applications of the partial path-recording oracle corresponding to the forward direction, and the $k$ in the inverse direction with a coincidence of the form $\langle x | \widetilde{D} | x \rangle$, which enables us to derive a condition in terms of $D$:
%
%
\begin{equation} \begin{aligned}
    \frac{1}{2} \, \mathrm{RHS} &\leq \sum_{i,i' \in [k]} \abs{\left\langle \Pi^{\text{eq}}_{i, i'} \right\rangle}^2 = \abs{\expect_C \left[ C^{\otimes k} W^{(k)} \left[ \sum_{i,i' \in [k]} \Pi^{\text{eq}}_{i, i'} \right] \widetilde{D} (W^\dagger)^{(k)}(C^\dagger)^{\otimes k} \right]\ket{\psi}}^2 \\
    &\leq \frac{1}{2^{n}} \sum_{i,i' \in [k]} \sum_{x_i, x_{i'} \in [N]} \abs{\langle x_{i'} | \widetilde{D} | x_i \rangle}^2 \delta(x_i, x_{i'}) \\
    &\leq \frac{k^2}{2^{n}} \sum_{x \in [N]} \langle xx |  \expect_{C'} \left[ (C')^{\otimes 2} ( D \otimes D^\dagger ) (C'^\dagger)^{\otimes 2} \right] | xx \rangle \\
    &\leq \frac{k^2}{2^{n} }\sum_{x \in [N]} \frac{1}{2^{2n}-1} \left[ \abs{\Tr D}^2 + \Tr \abs{D}^2 - \frac{1}{2^n} \Tr \abs{D}^2 - \frac{1}{2^n} \abs{\Tr D}^2 \right] \\
    &\leq \frac{2k^2}{2^{2n}} \abs{\Tr D}^2 + \frac{2k^2}{2^n},
    \end{aligned} \label{eq:D trace power moment bound}
    \end{equation}
where in the last step we have used the fact that $\Tr \abs{D}^2 = 2^n$. Plugging this back into Equation~\ref{eq:restrict random H moment Im W} yields
\begin{equation} \label{eq: D trace power moment bound}
    \norm{\Tr_{\gsL\gsR} |\calT^{\mathsf{inv}}_k \rangle \langle \calT^{\mathsf{inv}}_k | - \Tr_{\gsL\gsR} |\calT^{\widetilde{\mathsf{inv}}}_k \rangle \langle \calT^{\widetilde{\mathsf{inv}}}_k |}_1  \leq \frac{4k^2}{2^{2n}} \abs{\Tr D}^2 + \frac{4k^2}{2^n}.
\end{equation}
%
Expanding the action of the $W^k$ corresponding to the forward direction in $\ket{\calT^{\widetilde{\mathsf{inv}}}_k}$ yields
\begin{equation}
    \begin{aligned}
        \frac{1}{Z'} \sum_{(L,R) \in \calR^{2}_{\text{dist}}} & \ket{\{(x_{i'}', y_{i'}')\}_{{i'} \in [k]}}_\gsL \ket{\{(x_{i}, y_{i})\}_{{i} \in [k]}}_\gsR 
        \bigotimes_{j=1}^k \bra{x_j'} \widetilde{D} \ket{x_j} \cdot \expect_C \left[ C \ketbra{y_j'}{y_j} C^\dagger \right] \ket{\psi},
    \end{aligned}
\end{equation}
where we have included the normalization factor $Z' = (N!/(N-2k)!)^{-1/2}.$ In addition, we have that
\begin{equation}
    |\calT^{\widetilde{V}}_k \rangle = \frac{1}{Z} \sum_{\mathbf{y'} \in [N]^k_\mathrm{dist}} \ket{\{(y_{i}', y_{i})\}_{{i} \in [k]}}_\gsL \ket{\varnothing}_\gsR \bigotimes_{j=1}^k \expect_C \left[ C \ketbra{y_j'}{y_j} C^\dagger \right] \ket{\psi}.
\end{equation}
Computing the partial trace over the relation state registers and expanding up to leading order yields that the outputs of the two are indistinguishable up to
\begin{equation}
    \norm{\Tr_{\gsL \gsR} |\calT^{\widetilde{\mathsf{inv}}}_k \rangle \langle \calT^{\widetilde{\mathsf{inv}}}_k | - \Tr_{\gsL} |\calT^{\widetilde{V}}_k \rangle \langle \calT^{\widetilde{V}}_k |}_1 \leq 2\left[\left(\frac{2k^2}{2^{2n}} \abs{\Tr D}^2 + \frac{2k^2}{2^n}\right) + \left(\frac{2k^2}{2^n}\right) \right] + O\left(\frac{k^2}{2^{2n}}\right), 
    \label{eq:sum D trace power moment Vk bound}
\end{equation}
where we have again applied the triangle inequality to to bound the contributions to the norm from the $\abs{\langle x_{i}'| \widetilde{D} | x_i \rangle}^2$ and $\abs{\expect_C \left[ C \ketbra{y_j'}{y_j} C^\dagger \right] \ket{\psi}}^2$ terms. Finally, applying triangle inequality to Equations~\ref{eq:random H moment purify bound},~\ref{eq: D trace power moment bound},~\ref{eq:sum D trace power moment Vk bound} gives
%
%
\begin{equation}
\begin{gathered}
    \norm{\expect_{D \sim \calD} \expect_{U \sim \Haar} \left[\left(U  D U^\dagger\right)^{\otimes k} |\psi \rangle \langle \psi | \left(U D^\dagger U^\dagger\right)^{\otimes k} \right] - \Tr_{\gsL} |\calT^{\widetilde{V}}_k \rangle \langle \calT^{\widetilde{V}}_k |}_1 \\
    \leq \frac{144k(k+1)}{2^{n/8}} + \expect_{D \sim \calD} \left[ \frac{8k^2}{2^{2n}} \abs{\Tr D}^2 + \frac{12k^2}{2^n} + O\left(\frac{k^2}{N^{2}}\right) \right] \\
    = \frac{144k(k+1)}{N^{1/8}} + \frac{12k^2}{N} + \expect_{D \sim \calD} \left[ \frac{8k^2}{2^{2n}} \abs{\Tr D}^2 \right] + O\left(\frac{k^2}{N^{2}}\right), 
\end{gathered}
\end{equation}
which yields the stated bound up to leading order in $k/N$. \qed

\subsubsection{Gluing argument for nearly-local dynamics}

In order to prove Theorems~\ref{thm:generic-designs} and~\ref{thm:generic-PRUs}, we now present a gluing-style argument for neighboring patches of the form given in Definition~\ref{def:randomdynamics}. It is convenient to introduce a new isometry which \emph{expands} relation states of a given length, when restricted to a particular subspace.
\begin{definition}[Expansion map] \label{def:expand map}
    We construct a map which acts nontrivially on relation states in a single register $\gsL$ of length $k$ whose output is distinct on $S_1 \cup S_2$, so that $\abs{\Im_{S_1 \cup S_2}(\gsL)} = k$ for disjoint physical subsystems $S_1, S_2$ of size $\xi$ each. 
    
    On this subspace, the map ``expands'' the original relation state into superpositions over relation states on two intermediate registers which correspond to non-overlapping regions $S_1' \supseteq S_1, S_2' \supseteq S_2$, sandwiched by a pair of relation states $\gsL_0, \gsR_0$ corresponding to the original physical system:
    \begin{equation}
    \begin{aligned}
        \mathsf{Expand}^k_{\gsL \to \gsL_0 \gsR_0 \gsL_1 \gsL_2} = \sum_{\substack{\mathbf{w^1}, \mathbf{w^2}, \mathbf{z^1}, \mathbf{z^2} \in [2^\xi]^k_\mathrm{dist}}} \frac{1}{Z} & \left[ \ket{\{(x_i^1 x_i^2, w^1_i w^2_i)\}_{i \in [k]}}_{\gsL_0} \ket{\{(z^1_i z^2_i, y_i^1 y_i^2)\}_{i \in [k]}}_{\gsR_0} \right. \\ & \left.\ket{\{(w^1_i x^{\overline{1}}_i, z^1_i y^{\overline{1}}_i)\}_{i \in [k]}}_{\gsL_1} \ket{\{(w^2_i x^{\overline{2}}_i, z^2_i y^{\overline{2}}_i)\}_{i \in [k]}}_{\gsL_2} \right] \\ &\cdot \bra{\{(x_i,y_i)\}_{i \in [k]}}_{\gsL} \cdot \Pi^{\mathrm{dist}}_{S_1 S_2} ,
    \end{aligned}
    \end{equation}
    where $Z$ is the normalization factor $\left[ (2^\xi)!/(2^\xi-k)! \right]^{-2}$, and we have used the shorthand $x^1_i$ to denote the substring of x on $S_1$. This map acts as an isometry on the image of $\Pi^{\mathrm{dist}}_{S_1 S_2}$.
\end{definition}
We now present the formal statement of our bound for the ``glued'' dynamics:
\begin{lemma}[Gluing through conjugation] \label{lemma:glue patch moments}
    Consider the ensemble 
    \begin{equation}
        \calE' = \{U_{AA'} (U_{A\overline{A}} \otimes U_{A'\overline{A'}}) U_{AA'}^\dagger\},   
    \end{equation}
    where $U_{A\overline{A}}, U_{A'\overline{A'}}$ are drawn from from $\epsilon_A$- and $\epsilon_{A'}$-approximate $k$-designs, respectively, on the regions $A\overline{A}$ and $A'\overline{A'}$, and $U_{AA'}$ is a Haar-random transformation on the overlapping region $AA'$. In addition, suppose that $A$ and $A'$ consist of $\xi$ qubits each. Then $\calE'$ is an approximate $k$-design up to additive error
    \begin{equation}
        \norm{\Phi_{\calE'} - \Phi_H}_\diamond \leq \eps_A + \eps_{A'} + \frac{108 k(k+1)}{2^{\xi/8}} + \frac{18 k(k+1)}{2^{\xi/4}} + \frac{2k\sqrt{2}}{2^{\xi/2}} + \frac{k \sqrt{2}}{2^\xi}.
    \end{equation}
\end{lemma}
\proof Consider the ensemble given by
\begin{equation}
    \calE_0' = \{\left.U_{AA'} (U_{A\overline{A}} \otimes U_{A'\overline{A'}}) U_{AA'}^\dagger \, \right\rvert \, U_{A\overline{A}}, U_{A'\overline{A'}} \sim \Haar\},
\end{equation}
where the unitaries $U_{A\overline{A}}, U_{A'\overline{A'}}$ are still conjugated by a random overlapping unitary, but we have instead drawn them from the Haar measure on $A\overline{A}$ and $A'\overline{A'}$ rather than from an approximate $k$-design. By assumption, we have that
\begin{equation}
    {\norm{\Phi_{\calE'} - \Phi_{\calE'_0}}_\diamond \leq \eps_A + \eps_{A'}}. \label{eq:nonadaptive glue design to haar}
\end{equation}
We now apply Theorem~\ref{theorem:haar-cho-strong} in order to construct an approximate purification of the moments of $\calE_0'$. For any input state $\ket{\psi}$ on $nk + m$ qubits, the output is approximated by tracing over the relation state registers of
\begin{equation}
    \ket{\calT^{\mathsf{conj}}_k} = V^{(k)}_{AA'} \left(V_{A\overline{A}} \otimes V_{A'\overline{A'}} \right)^{(k)} (V^\dagger_{AA'})^{(k)} \ket{\psi},
\end{equation}
where we have instantiated multiple path-recording oracles $V_S$ with separate relation state registers corresponding to subsystem $S$. We obtain an exponentially small error in trace distance:
\begin{equation}
    \norm{\Tr_{\gsL_{A\overline{A}} \gsL_{A'\overline{A'}} \gsL_{\overline{A} \overline{A'}} \gsR_{\overline{A} \overline{A'}}} \ketbra*{\calT^{\mathsf{conj}}_k} - \expect_{U \sim \calE_0'} U^{\otimes k} \ketbra*{\psi} (U^\dagger)^{\otimes k}}_1 \leq \frac{108 k(k+1)}{2^{\xi/8}}. \label{eq:nonadapative glue purify ensemble}
\end{equation}
We will consider an intermediate comparison to the approximate purification of the Haar moment $\ket{\calT^V_k}$. To do so, we examine the restriction of $\ket{\calT^{\mathsf{conj}}_k}$ to a particular subspace for which we can construct a partial isometry mapping to the image of $V^{(k)}$. In particular, we insert a \emph{doubly} distinct projection before and after the applications of $V_{A\overline{A}} \otimes V_{A\overline{A'}}$: 
\begin{equation}
    \ket{\calT^{\widetilde{\mathsf{conj}}}_k} = V^{(k)}_{AA'} \left[\Pi_{A}^\mathrm{dist} \otimes \Pi_{A'}^\mathrm{dist} \right] \left(V_{A\overline{A}} \otimes V_{A'\overline{A'}} \right)^{(k)} \left[\Pi_{A}^\mathrm{dist} \otimes \Pi_{A'}^\mathrm{dist} \right] (V^\dagger_{AA'})^{(k)} \ket{\psi},
\end{equation}
Since these projectors are all diagonal in the computational basis, it is clear that the support of $\ket{\calT^{\widetilde{\mathsf{conj}}}_k}$ is contained in the support of $\ket{\calT^{\mathsf{conj}}_k}$. The error incurred by enforcing this restriction is therefore given by the magnitude of the remaining portion not contained in the support of $\ket{\calT^{\widetilde{\mathsf{conj}}}_k}$. 
\begin{equation} \label{eq:nonadaptive glue restrict purification}
\begin{aligned}
    &\, \norm{\Tr_{\gsL_{A\overline{A}} \gsL_{A'\overline{A'}} \gsL_{\overline{A} \overline{A'}} \gsR_{\overline{A} \overline{A'}}} \ketbra*{\calT^{\mathsf{conj}}_k} - \Tr_{\gsL_{A\overline{A}} \gsL_{A'\overline{A'}} \gsL_{\overline{A} \overline{A'}} \gsR_{\overline{A} \overline{A'}}} \ketbra*{\calT^{\widetilde{\mathsf{conj}}}_k}}_1 \\
    \leq &\, \norm{V^{(k)}_{AA'} \left[\mathbbm{1} -  \Pi_{A}^\mathrm{dist} \otimes \Pi_{A'}^\mathrm{dist} \right] \left(V_{A\overline{A}} \otimes V_{A'\overline{A'}} \right)^{(k)} \left[\Pi_{A}^\mathrm{dist} \otimes \Pi_{A'}^\mathrm{dist} \right] (V^\dagger_{AA'})^{(k)} \ket{\psi}}_2 \\ 
    &\, + \norm{V^{(k)}_{AA'} \left[\Pi_{A}^\mathrm{dist} \otimes \Pi_{A'}^\mathrm{dist} \right] \left(V_{A\overline{A}} \otimes V_{A'\overline{A'}} \right)^{(k)} \left[\mathbbm{1} - \Pi_{A}^\mathrm{dist} \otimes \Pi_{A'}^\mathrm{dist} \right] (V^\dagger_{AA'})^{(k)} \ket{\psi}}_2 \\
    = &\, \norm{V^{(k)}_{AA'} \left[\mathbbm{1} - \Pi_{A}^\mathrm{dist} \otimes \Pi_{A'}^\mathrm{dist} \right] \left[\Pi_{A \overline{A}}^\mathrm{dist} \otimes \Pi_{A' \overline{A'}}^\mathrm{dist} \right] \left(V_{A\overline{A}} \otimes V_{A'\overline{A'}} \right)^{(k)} \left[\Pi_{A}^\mathrm{dist} \otimes \Pi_{A'}^\mathrm{dist} \right] (V^\dagger_{AA'})^{(k)} \ket{\psi}}_2 \\ 
    &\, + \norm{V^{(k)}_{AA'} \left[\Pi_{A}^\mathrm{dist} \otimes \Pi_{A'}^\mathrm{dist} \right] \left(V_{A\overline{A}} \otimes V_{A'\overline{A'}} \right)^{(k)} \left[\mathbbm{1} - \Pi_{A}^\mathrm{dist} \otimes \Pi_{A'}^\mathrm{dist} \Pi_{AA'}^\mathrm{dist} \right] (V^\dagger_{AA'})^{(k)} \ket{\psi}}_2 
    \leq \frac{2k\sqrt{2}}{2^{\xi/2}},
\end{aligned}
\end{equation}
where in the second step we have used the fact that the output of each $V_S^{(k)}$ corresponding to the physical subsystem $S$ is distinct on $S$, then applied a standard counting argument for the portion of the distinct subspace which is contained on a local distinct subspace as given in Definition~\ref{def:restrict local distinct}. In addition, we consider the corresponding restriction $\Pi^{\mathrm{dist}}_{AA'}$ on the output of $V^{(k)}_{A\overline{A} A' \overline{A'}}$, and similarly bound its distance from the full output of $\ket{\calT^V_k}$: 
\begin{equation} \label{eq:nonadaptive glue restrict vk}
    \norm{\Tr_{\gsL} \ketbra*{\calT^V_k} - \Tr_{\gsL} \left[ \Pi^{\mathrm{dist}}_{AA'} \ketbra*{\calT^V_k} \right]}_1 \leq \frac{k \sqrt{2}}{2^{\xi}}.
\end{equation}
By inspection, we have that the $\mathsf{Expand}$ isometry maps the restriction $\Pi^{\mathrm{dist}}_{AA'} \ket{\calT^V_k}$ to the output of $\ket{\calT^{\widetilde{\mathsf{conj}}}_k}$:
\begin{equation} 
    \mathsf{Expand}^k_{\gsL \to \gsL_{A\overline{A}} \gsL_{A'\overline{A'}} \gsL_{\overline{A} \overline{A'}} \gsR_{\overline{A} \overline{A'}}} \left[ \Pi^{\mathrm{dist}}_{AA'} \ketbra*{\calT^V_k} \right] = \ket{\calT^{\widetilde{\mathsf{conj}}}_k}.
\end{equation}
Since this map is isometric and acts as a change of basis on the relation state registers, we have
\begin{equation} \label{eq:nonadaptive glue expand}
    \Tr_{\gsL} \left[ \Pi^{\mathrm{dist}}_{AA'} \ket{\calT^V_k} \right] = \Tr_{\gsL_{A\overline{A}} \gsL_{A'\overline{A'}} \gsL_{\overline{A} \overline{A'}} \gsR_{\overline{A} \overline{A'}}} \ketbra*{\calT^{\widetilde{\mathsf{conj}}}_k}.
\end{equation}
Finally, we apply Theorem~\ref{theorem:haar-cho-strong} to obtain
\begin{equation} \label{eq:nonadaptive glue vk to haar}
    \norm{\Tr_{\gsL} \ketbra*{\calT^V_k} - \expect_{U \sim \Haar} U^{\otimes k} \ketbra*{\psi} (U^\dagger)^{\otimes k}}_1 \leq \frac{18k(k+1)}{2^{\xi/4}}.
\end{equation}
We complete the proof by applying the triangle inequality to the results of Equations~\ref{eq:nonadaptive glue design to haar},~\ref{eq:nonadapative glue purify ensemble},~\ref{eq:nonadaptive glue restrict purification},~\ref{eq:nonadaptive glue expand},~\ref{eq:nonadaptive glue restrict vk}, and~\ref{eq:nonadaptive glue vk to haar}, which yields the stated bound. \qed
\\

We now restate the main results of this section, which states that the moments of the nearly-local dynamics are close to those of the Haar ensemble, up to any polynomial order.
\begin{theorem}[Theorem~\ref{thm:generic-designs} in main text] \label{thm:generic-designs-restate}
    Suppose $H$ is drawn from a variant of the random two-layer Hamiltonian ensemble in Theorem~\ref{thm:designs} in which each random diagonal term, $\sum_z J^i_z \dyad{z}_i$, is replaced with any fixed Hamiltonian $H_i$.
    The resulting time-evolution $U = e^{-iHt}$ forms an additive-error $\varepsilon$-approximate unitary $k$-design for any $H_i$ and any time $t$ such that $|\! \Tr(e^{-iH_i t})|^2 = o(\varepsilon/nk^2)$ for all $i$.
\end{theorem}
\begin{theorem}[Theorem~\ref{thm:generic-PRUs} in main text] \label{thm:generic-PRUs-restate}
    Suppose $H$ is drawn from a variant of the random two-layer Hamiltonian ensemble in Theorem~\ref{thm:designs} in which $\sum_z J^i_z \dyad{z}_i$ is replaced with any fixed Hamiltonian $H_i$.
    The resulting time-evolution $U = e^{-iHt}$ forms a PRU with non-adaptive security for any $H_i$ and any time $t$ such that $|\! \Tr(e^{-iH_i t})|^2 = o(1/\poly n)$ for all $i$.
\end{theorem}
\begin{proof}[Proof of Theorems~\ref{thm:generic-designs-restate} and~\ref{thm:generic-PRUs-restate}]
    We first observe that $H_i$ admits a fixed diagonalization $H_i = U_0 D_0 U_0^\dagger$. By assumption, the Hamiltonian ensemble is invariant under change of basis on each patch. Both statements then follow immediately from substituting the result of Proposition~\ref{prop:random H moments} into each patch for the ensemble of interest, then applying the result of Lemma~\ref{lemma:glue patch moments} to each instance of conjugation by overlapping unitaries, up to $n/\xi$. Setting $\xi = \poly\log(n)$ yields the stated result.
\end{proof}

\subsubsection{Designs via local spectral distribution}

While our bound above holds for $D$ sampled from arbitrary spectral distributions, we are often interested in systems whose distributions have particular properties, such as being generated by an extensive number of local degrees of freedom. We remark that it is possible to obtain superpolynomial-order designs up to superpolynomially small order even for distributions that are generated by single-qubit operators.
\begin{fact} \label{fact: local phases exp}
    Suppose $D$ is sampled from a distribution which is constructed by taking a product of an extensive number of i.i.d. local random phases $D_i \sim \calD_0$. Then 
    \begin{equation}
        \expect_{D \sim \calD} \left[ \abs{\Tr D}^2 \right] = \left(\expect_{D' \sim \calD_0} \abs{\Tr D'}^2 \right)^{n/n_0},
    \end{equation}
    where $n_0$ is the size of the physical subsystem corresponding to $\calD_0$.
\end{fact}
In particular, this implies that the bound of Proposition~\ref{prop:random H moments} is exponentially small in $n$ for models in which the energy spectrum is generated by single-qubit operators with eigenvalues $E_0$ and $E_1$, and the evolution time is concentrated around
\begin{equation} \label{eq:constant time scale}
    t_0 = \pi\hbar/\Delta E = \pi\hbar/\abs{E_0 - E_1}.
\end{equation}
\begin{corollary}[Statistics of random eigenbasis dynamics are close to Haar-random] \label{cor:random moments}
    For any $k = 2^{o(n)}$, the random eigenbasis Hamiltonian dynamics with phases generated by single-qubit operators with constant energy gap $\Delta E$ forms an approximate $k$-design at constant time scales $t \sim G(t_0,t_0/8)$, up to additive error
    \begin{equation}
        {\norm{\Phi_{\calE_\mathsf{inv}} - \Phi_H}_\diamond} \leq \frac{162k(k+1)}{N^{1/8}} + \frac{12k^2}{N} + \frac{8k^2}{c^{n}} + O\left(\frac{k^2}{N^{2}}\right),
    \end{equation}
    where $t_0$ is given by Equation~\ref{eq:constant time scale}, $G(\mu, \sigma)$ denotes the normal distribution, and $c$ is a positive constant greater than one.
\end{corollary}
\proof We bound the trace moments via Fact~\ref{fact: local phases exp} 
\begin{equation}
    \begin{aligned}
        \expect_{D \sim \calD} \left[ \abs{\Tr D}^2 \right] =&\, 
        \left[\int_{0}^{2\pi} d\theta \left[\left(\cos(x) + 1\right)^2 + \sin^2(\theta) \right] \left[ \frac{1}{\sqrt{2\pi (\pi/8)^2}} e^{-(\theta - \pi)^2/2(\pi/8)^2} \right] \right]^n < 1. \\
    \end{aligned}
\end{equation}
The result then follows immediately from an application of the triangle inequality to the statements of Proposition~\ref{prop:random H moments} and Theorem~\ref{theorem:haar-cho-strong}, and using the fact that $V$ satisfies an approximate two-sided unitary invariance property, as described in Claim~\ref{claim:two-sided-invariance}. \qed
\begin{proposition}[Statistics of nearly-local dynamics with local spectrum are close to Haar-random] \label{thm:quasi-local moments formal}
    When $\calD$ is generated by single-qubit operators with constant energy gap $\Delta E$ and the time scale is taken to be $t \sim G(t_0,t_0/8)$, the ensemble of nearly-local Hamiltonians is an approximate $k$-design up to additive error
    \begin{equation}
        \epsilon = \frac{270nk(k+1)}{2^{\xi/8}\xi} + \frac{8nk^2}{c^{\xi}\xi} + O\left(\frac{k^2}{2^{\xi/4} \xi}\right),
    \end{equation}
    where $c$ is a positive constant greater than one, and we have left out subleading terms in $k/{2^\xi}$. In particular, when $\xi = O(\poly\log n)$, this ensemble forms an approximate design for all $k = \poly(n)$.
\end{proposition}
\begin{proof}
    This follows from substituting the result of Corollary~\ref{cor:random moments} into each patch for the ensemble of interest, then applying the result of Lemma~\ref{lemma:glue patch moments} to each instance of conjugation by overlapping unitaries, up to $n/\xi$. Setting $\xi = \poly\log(n)$ yields the stated result.
\end{proof}

\subsubsection{A simple alternative construction and proof}

We conclude our discussion of the non-adaptive setting by introducing an alternative random Hamiltonian time-evolution ensemble and providing an especially short proof of its indistinguishability from Haar-random in any non-adaptive quantum experiment.
We consider the random unitary ensemble,
\begin{equation}
    \mathcal{E}_{\mathbf{F}} = \left( \otimes_{i \in \text{even}} U_{i,i+1} \right)^\dagger  \cdot \left( \otimes_{i \in \text{odd}} F_{i,i+1} \right) \cdot \left( \otimes_{i \in \text{even}} U_{i,i+1} \right),
\end{equation}
where each $U_{i,i+1}$ is a strong PRU on $2\xi = \omega(\log n)$ qubits with security against any $\poly n$-time quantum adversary, and each $F_{i,i+1}$ is a PRF on $2\xi = \omega(\log n)$ qubits with security against any $\poly n$-time quantum adversary. While the ensemble of Hamiltonians and evolution time scale is not made explicit in the definition of $\calE_{\mathbf{F}}$, for any distribution over evolution time which can be efficiently specified, it is also possible to extract the corresponding ensemble of Hamiltonians. Moreover, fixing the time scale to be concentrated around some constant $t_0 = O(1)$ yields a well-defined ensemble of Hamiltonians whose time dynamics can be efficiently computed at any time scale.

We remark that in this ensemble, the basis transformation is itself non-entangling across several cuts of the system. However, it is still possible to obtain approximate designs up to superpolynomial order and error due to the action of the pseudorandom phases:
\begin{theorem}
    The time evolution of the ensemble of Hamiltonians described via $\mathcal{E}_{\mathbf{F}}$ with non-entangling basis transformations and spectral distribution which generates pseudorandom phases on nearly-local patches forms an approximate $k$-design up to additive error
    \begin{equation}
        \epsilon = \frac{36nk(k+1)}{2^{\xi/8} \xi} + \frac{4nk}{2^{\xi/2} \xi} + O\left(\frac{k^2}{N^{1/8}}\right).
    \end{equation}
\end{theorem}


%
\begin{proof}
    Consider any input state $\ket{\psi}$ over $nk + m$ qubits. For simplicity, we will assume that $n = n'\xi,$ where $\xi = \omega(\log n)$.
    Recall the path-recording oracle $V_i$ for simulating Haar-random unitaries $U_i$ under both forward and inverse queries and the purification $O_{F_{i, i+1}}$ of Haar-random diagonal unitaries $F_{i, i+1}$ as in \cite{ma2025construct}. As with the construction of $V^{(k)}$ given in Definition~\ref{def:parallel oracle}, we can consider a parallel version of the random phase oracle $O_{F_{i, i+1}}^{(k)}$ which acts equivalently to the original oracle in a restricted experiment where each of the $k$ queries do not use any adaptive postprocessing on the previous queries.
    We define the following pure states:
    \begin{align}
        &\ket{\calT^{\mathbf{V} \mathbf{F} \mathbf{V}^\dagger}_k} = \left[ \bigotimes_{i=1}^{n'} V^{(k)}_i \bigotimes_{i \in \text{even}} O_{F_{i, i+1}}^{(k)} \bigotimes_{i \in \text{odd}} O_{F_{i, i+1}}^{(k)} \bigotimes_{i=1}^{n'} (V^\dagger_i)^{(k)} \right] \ket{\psi}_\gsA \ket{\varnothing^{n'}}_{\bm{\gsL}} \ket{\varnothing^{n'}}_{\bm{\gsR}} \ket{0^{n'}}_{\bm{\gsE}},\\
        &\ket{\calT^{VFV^\dagger}_k} = \left[ V^{(k)} O_{F}^{(k)} (V^\dagger)^{(k)} \right] \ket{\psi}_\gsA \ket{\varnothing^{n'}}_{{\gsL}} \ket{\varnothing^{n'}}_{{\gsR}} \ket{0^{n'}}_{{\gsE}}.
    \end{align}
    Applying Theorem~\ref{theorem:haar-cho-strong} to each patch and using triangle inequality yields 
    \begin{equation}
        \norm{\E_{U \sim \mathcal{E}_{\mathbf{F}}} \left[ U^{\otimes k} \ketbra{\psi}{\psi} (U^\dagger)^{\otimes k}  \right] - \Tr_{\bm{\gsL \gsR \gsE}} \ketbra*{\calT^{\mathbf{V} \mathbf{F} \mathbf{V}^\dagger}_k}}_1 \leq \frac{36nk(k+1)}{2^{\xi/8} \xi}.
    \end{equation}
    Substituting in the definitions of $V_i$ and $O^{F_{i, i+1}}$ from \cite{ma2025construct}, then using a similar argument as \cite{cui2025unitary} enables us to compare $\ket{\calT^{\mathbf{V} \mathbf{F} \mathbf{V}^\dagger}_k}$ to an ensemble constructed via random unitary and phase transformations on the entire system. 
    In particular, we observe that the action of the $\otimes_{i \in [n']} V^{(k)}_i$ and $\otimes_{i \in \text{even}} O_{F_{i, i+1}}^{(k)} \otimes_{i \in \text{odd}} O_{F_{i, i+1}}^{(k)}$ is equivalent to that of $V^{(k)} \cdot \Pi^{\mathrm{dist}}_{\mathrm{loc}}$, and $O_{F}^{(k)} \cdot \Pi^{\mathrm{dist}}_{\mathrm{loc}}$, respectively, after tracing over the purifying register, where the projector onto the local distinct subspace is given via Definition~\ref{def:restrict local distinct} and $F$ is drawn from a family of PRFs on $n$ bits. An application of the gentle measurement lemma then yields
    \begin{equation}
        \norm{\Tr_{\bm{\gsL \gsR \gsE}} \ketbra*{\calT^{\mathbf{V} \mathbf{F} \mathbf{V}^\dagger}_k} - \Tr_{{\gsL \gsR \gsE}} \ketbra*{\calT^{{V} {F} {V}^\dagger}_k}}_1 \leq \frac{4nk}{2^{\xi/2} \xi},
    \end{equation}
    where we have used the fact that the restrictions can be applied after all oracle queries. 
    Applying triangle inequality and collecting error terms yields
    \begin{equation}
        \norm{\E_{U \sim \mathcal{E}_{\mathbf{F}}} \left[\left(U\right)^{\otimes k} |\psi \rangle \langle \psi | \left(U^\dagger\right)^{\otimes k} \right] - \Tr_{\gsL} \ketbra*{\calT^{{V} {F} {V}^\dagger}_k}}_1 \leq \frac{36nk(k+1)}{2^{\xi/8} \xi} + \frac{4nk}{2^{\xi/2} \xi}.
    \end{equation}
    We can now analyze $\ket{\calT^{{V} {F} {V}^\dagger}_k}$ following the approach of Proposition~\ref{prop:random H moments}. We observe that by assumption, the value of the trace moments are bounded via
    \begin{equation}
        \abs{\expect_{F} \left[\frac{1}{N} \abs{\Tr F}^2\right] - \expect_{U \sim \Haar} \left[\frac{1}{N} \abs{\Tr U}^2\right]} \leq \frac{1}{\omega(\poly n)},
    \end{equation}
    since the normalized moments correspond to physical measurements. 
    We therefore have that substituting the intermediate bounds from Proposition~\ref{prop:random H moments} and Theorem~\ref{theorem:haar-cho-strong} yields
    \begin{equation}
        \norm{\Tr_{\gsL} \ketbra*{\calT^{{V} {F} {V}^\dagger}_k} - \Tr_{\gsL} |\calT^{\widetilde{V}}_k \rangle \langle \calT^{\widetilde{V}}_k |}_1 \leq \frac{36k(k+1)}{N^{1/8}} + \frac{12k^2}{N} + \frac{8k^2}{N}.
    \end{equation}
    As a result, by a final triangle inequality, we have that
    \begin{equation}
        \norm{\E_{U \sim \mathcal{E}_{\mathbf{F}}} \left[ U^{\otimes k} \ketbra{\psi}{\psi} (U^\dagger)^{\otimes k}  \right] - \E_{U \sim \Haar} \left[ U^{\otimes k} \ketbra{\psi}{\psi} (U^\dagger)^{\otimes k} \right] }_{1} \leq \frac{36nk(k+1)}{2^{\xi/8} \xi} + \frac{4nk}{2^{\xi/2} \xi} + O\left(\frac{k^2}{N^{1/8}}\right),
    \end{equation}
    where we have again used the bound from Theorem~\ref{theorem:haar-cho-strong} to compare $|\calT^{\widetilde{V}}_k \rangle$ when traced over the purifying register, to the action of a Haar-random unitary. This yields the stated result.
\end{proof}
%


\subsection{Proof of Theorems~\ref{thm:designs} and~\ref{thm:PRUs}: Indistinguishability in adaptive quantum experiments} \label{app:adaptive}

Having demonstrated our simpler proofs in the non-adaptive setting, we now turn to the adaptive setting and provide the complete proof of Theorems~\ref{thm:designs} and~\ref{thm:PRUs}.
We consider the following random unitary ensemble, which is generated by evolving the random Hamiltonian ensemble described in the main text for an evolution time $t= \pi$, 
\begin{equation}
    U = \left( \bigotimes_{i \in \text{even}} U_{i,i+1}^\dagger \right) \left( \bigotimes_{i \in \text{odd}} U_{i,i+1}^\dagger \right) \left( \bigotimes_{i \in \text{odd}} F_{i,i+1} \right) \left( \bigotimes_{i \in \text{odd}} U_{i,i+1} \right) \left( \bigotimes_{i \in \text{even}} U_{i,i+1} \right).
\end{equation}
In Theorem~\ref{thm:designs}, each small random unitary $U_{i,i+1}$ and each small PRF $F_{i,i+1}$ are chosen to be indistinguishable from a Haar-random unitary and a truly random function by any $k$-query quantum experiment.
In Theorem~\ref{thm:PRUs}, the are chosen to be indistinguishable by any polynomial-time quantum experiment.

As described in the main text, our proof of Theorems~\ref{thm:designs} and~\ref{thm:PRUs} follows from a ``gluing'' argument.
To this end, we have the following two lemmas, which we prove in the subsequent subsections.
\begin{lemma}[Conjugated random phases are random unitaries~\cite{gu2024simulating}] \label{lemma:conj-phase}
    Let $\mathcal{E}$ be equal to the product $U^\dagger F U$, where $U$ is Haar-random and $F$ is a random continuous phase gate.
    Then $\mathcal{E}$ is indistinguishable from a Haar-random unitary up to measurable error $\mathcal{O}(k/N^{1/4})$.
\end{lemma}
\noindent The lemma is not stated explicitly in~\cite{gu2024simulating}, but is derived in several steps spread between different proofs of their work.
We provide a concise step-by-step summary of the proof in Section~\ref{sec:conj-phase} below.
\begin{lemma}[Gluing via conjugation] \label{lemma:gluing}
    Let $\mathsf{a}$, $\mathsf{b}$, $\mathsf{c}$, $\mathsf{d}$ be four subsystems.
    Consider the ensemble $\mathcal{E}$ on $\mathsf{abcd}$ given by the product of Haar-random unitaries, $U_{\mathsf{bc}}^\dagger (U_{\mathsf{ab}} \otimes U_{\mathsf{cd}}) U_{\mathsf{bc}}$.
    The ensemble $\mathcal{E}$ is indistinguishable from a Haar-random unitary $U_\mathsf{abcd}$ up to measurable error 
    $\mathcal{O}(t^2/N_{\mathsf b}^{1/2})+\mathcal{O}(t^2/N_{\mathsf c}^{1/2}) + \mathcal{O}(t^2/N_{\mathsf{bc}}^{1/8})$.
\end{lemma}
\noindent The proof of this lemma is significantly more involved and is completed via the path-recording framework in Section~\ref{sec:gluing}.

Our proof of Theorems~\ref{thm:designs} and~\ref{thm:PRUs} follows simply from Lemma~\ref{lemma:conj-phase} and Lemma~\ref{lemma:gluing}.
\begin{proof}[Proof of Theorems~\ref{thm:designs} and~\ref{thm:PRUs}]
    By definition, each strong PRU is indistinguishable from a Haar-random unitary by any $\poly n$-time quantum experiments, and each PRF is indistinguishable from a uniform continuous random phase by any $\poly n$-time quantum experiment.
    Hence, we will assume $U_{i,i+1}$ and $F_{i,i+1}$ are truly random from hereon.

    From Lemma~\ref{lemma:conj-phase}, each product $U_{i,i+1}^\dagger F_{i,i+1} U_{i,i+1}$ for $i$ odd is indistinguishable from a Haar-random unitary $U'_{i,i+1}$ up to measurable error $\mathcal{O}(k/2^{\xi/2})$ in any $k$-query quantum experiment.
    There are $n/2\xi$ such $i$, which leads to a total measurable error $\mathcal{O}((n/2\xi)k/2^{\xi/2})$.
    Then, applying Lemma~\ref{lemma:gluing} $n/2\xi$ times sequentially from left to right, we find that the entire product, $(\otimes_{i \in \text{even}} U_{i,i+1})^\dagger (\otimes_{i \in \text{even}} F_{i,i+1}) (\otimes_{i \in \text{even}} U_{i,i+1})$, is indistinguishable from a Haar-random unitary on all $n$ qubits up to measurable error $\mathcal{O}(k^2/2^{\xi/4})$ in any $k$-query quantum experiment.
    Hence, the random nearly-local Hamiltonian time-evolution ensemble is indistinguishable from a Haar-random unitary up to measurable error $\mathcal{O}(k^2/2^{\xi/4})$.
    For any $\xi = \omega(\log n)$, the measurable error is super-polynomially small in $n$ for any $k = \poly n$, and hence the ensemble forms a PRU with security against any $\poly n$-time quantum adversary.
\end{proof}

\subsubsection{Proof of Lemma~\ref{lemma:conj-phase}: Conjugated random phases are random unitaries} \label{sec:conj-phase}

    We provide a concise summary of the proof in~\cite{gu2024simulating}, and refer to~\cite{gu2024simulating} for further details.
    We can diagonalize a Haar-random unitary $U'$ as $U' = U^\dagger F_{CUE} U$, where $U$ is Haar-random and $F_{CUE}$ is randomly drawn from the CUE spectral distribution~\cite{gu2024simulating}.
    Note that both the probability distribution $p$ of $N$ uniform random phases and the probability distribution $p_{CUE}$ of $N$ CUE-distributed random phases are invariant under any permutation of the $N$ bitstrings.
    From this property, Eq.~(B.34) of~\cite{gu2024simulating} shows that the measurable error of any quantum experiment that queries $F$ versus $F_{CUE}$ up to $k$ times, is bounded above by the total variation distance between the $2k$-body marginal $p^{(2k)}$ of $p$ and the $2k$-body marginal $p_{CUE}^{(2k)}$ of $p_{CUE}$.
    From Lemma~23 of~\cite{gu2024simulating}, the $2k$-body marginals of the uniform and CUE distributions are close up to total variation distance $\mathcal{O}(k/N^{1/4})$.
    Hence, the measurable error between $F$ and $F_{CUE}$ is at most $\mathcal{O}(k/N^{1/4})$.
    This immediately implies that the measurable error between $U^\dagger F U$ and $U' = U^\dagger F_{CUE} U$ is at most $\mathcal{O}(k/N^{1/4})$ as well. \qed 

\subsubsection{Proof of Lemma~\ref{lemma:gluing}: Gluing random unitaries via conjugation} \label{sec:gluing}

%

%

    Let $U \equiv U_{\mathsf{bc}}^\dagger (U_{\mathsf{ab}} \otimes U_{\mathsf{cd}}) U_{\mathsf{bc}}$ be the unitary of interest, and $U' \equiv U_{\mathsf{bc}}^\dagger (U_{\mathsf{abcd}}) U_{\mathsf{bc}} = U_{\mathsf{abcd}}$ a Haar-random unitary.
    In the latter expression, we have used the invariance of the Haar measure under any unitary rotation to eliminate $U_{\mathsf{bc}}, U_{\mathsf{bc}}^\dagger$.

    From~\cite{ma2025construct}, we can replace each query to $U$ with the product of path-recording oracles,
    $V \equiv V_{\mathsf{bc}}^\dagger (V_{\mathsf{ab}} \otimes V_{\mathsf{cd}}) V_{\mathsf{bc}}$,
    acting on auxiliary spaces $\mathsf{L}_{\mathsf{ab}}$, $\mathsf{L}_{\mathsf{cd}}$,
    $\mathsf{L}_{\mathsf{bc}}$,
    $\mathsf{R}_{\mathsf{bc}}$.
    In particular, we take $V_{\mathsf{bc}}$, $V_{\mathsf{bc}}^\dagger$ identical to~\cite{ma2025construct}.
    From~\cite{ma2025construct}, this incurs a measurable error $\mathcal{O}(t^2/N_{\mathsf{bc}}^{1/8})$.
    On the other hand, we make two modifications to the path-recording oracle $V_{\mathsf{ab}}$.
    First, following Appendix~B of~\cite{ma2025construct}, we take $V_{\mathsf{ab}}$ to output bitstrings $y_{\mathsf{ab}}$ that are locally distinct on $\mathsf{b}$, and $V_{\mathsf{cd}}$ to output bitstrings $y_{\mathsf{cd}}$ that are locally distinct on $\mathsf{c}$.
    This incurs a measurable error $\mathcal{O}(t^2/N_{\mathsf b})+\mathcal{O}(t^2/N_{\mathsf c})$ following Appendix~B of~\cite{ma2025construct}.
    Second, we further enforce that the output bitstrings are locally distinct on $\mathsf{b}$ from all previous \emph{inputs} to $V_{\mathsf{ab}}$, and similar for $\mathsf{c}$ and $V_{\mathsf{cd}}$.
    Inserting such a projector reduces the normalization of the state by at most $\mathcal{O}(t/N_{\mathsf{b}})$ per application, and hence by at most $\mathcal{O}(t^2/N_{\mathsf{b}})$ after $t$ applications (and similar for $\mathsf{c}$).
    From the sequential gentle measurement lemma~\cite{ma2025construct}, this incurs a total measurable error of at most $\mathcal{O}(t^2/N_{\mathsf{b}}^{1/2}) + \mathcal{O}(t^2/N_{\mathsf{c}}^{1/2})$.
    
    We proceed in a similar manner for $U'$.
    We replace each query to $U'$ with the product, 
    $V' \equiv V_{\mathsf{bc}}^\dagger (V_{\mathsf{abcd}}) V_{\mathsf{bc}}$,
    acting on auxiliary spaces $\mathsf{L}_{\mathsf{abcd}}$,
    $\mathsf{L}_{\mathsf{bc}}$,
    $\mathsf{R}_{\mathsf{bc}}$.
    We take $V_{\mathsf{bc}}$, $V_{\mathsf{bc}}^\dagger$ as in~\cite{ma2025construct} as before.
    We then take $V_{\mathsf{abcd}}$ as in the previous paragraph, to output bitstrings $y_{\mathsf{abcd}}$ that are locally distinct on both $\mathsf{b}$ and $\mathsf{c}$, and which are also subsequently projected to also be locally distinct from all previous inputs to $V_{\mathsf{abcd}}$ on $\mathsf b$ and $\mathsf c$.
    In total, these replacements accrue a measurable error $\mathcal{O}(t^2/N_{\mathsf{bc}}^{1/8})+\mathcal{O}(t^2/N_{\mathsf b}^{1/2})+\mathcal{O}(t^2/N_{\mathsf c}^{1/2})$, which has the same scaling as measurable error accrued in the previous paragraph.

    From these definitions, for both $V$ and $V'$, the input to $V_{\mathsf{bc}}^\dagger$ is always distinct from all previous outputs of $V_{\mathsf{bc}}$.
    This follows from our projection on the output of $V_{\mathsf{ab}} \otimes V_{\mathsf{cd}}$ (or $V_{\mathsf{abcd}}$) onto bitstrings that are locally distinct from all previous inputs; every output of $V_{\mathsf{bc}}$ is an input to the proceeding $V_{\mathsf{ab}} \otimes V_{\mathsf{cd}}$ (or $V_{\mathsf{abcd}}$).
    The input to $V_{\mathsf{bc}}^\dagger$ is also always distinct from all previous inputs to $V_{\mathsf{bc}}^\dagger$, since the output of $V_{\mathsf{ab}} \otimes V_{\mathsf{cd}}$ (or $V_{\mathsf{abcd}}$) is locally distinct from all previous such outputs.
    Hence, we can replace $V_{\mathsf{bc}}^\dagger$ with its restriction on the distinct subspace, $W_{\mathsf{bc},R}$~\cite{ma2025construct}.
    This guarantees that $V_{\mathsf{bc}}^\dagger$ always creates a new entry in the $\mathsf{R}_{\mathsf{bc}}$ register whenever it is queried.

    From this property, it follows that the states on the auxiliary registers when we apply $U$ correspond to ``chains'' of the form,
    \begin{equation} \label{eq: U aux}
    \begin{split}
        & \ket{
        \cup_{i=1}^m
        \{
        (x^i_{\mathsf b}x^i_{\mathsf c},y^{i,0}_{\mathsf b}y^{i,0}_{\mathsf c})
        \}
        }_{\mathsf{L}_{\mathsf{bc}}} \\
        & \ket{
        \cup_{i=1}^m
        \{
        (x^{i,1}_{\mathsf a}y^{i,0}_{\mathsf b},
        z^{i,1}_{\mathsf a}y^{i,1}_{\mathsf b}),
        \ldots,
        (x^{i,t_i}_{\mathsf a} y^{i,t_i-1}_{\mathsf b},
        z^{i,t_i}_{\mathsf a} y^{i,t_i}_{\mathsf b})
        \}}_{\mathsf{L}_{\mathsf{ab}}} \\
        & \ket{
        \cup_{i=1}^m
        \{
        (y^{i,0}_{\mathsf c} x^{i,1}_{\mathsf d},
        y^{i,1}_{\mathsf c} z^{i,1}_{\mathsf d}),
        \ldots,
        (y^{i,t_i-1}_{\mathsf c} x^{i,t_i}_{\mathsf d} ,
        y^{i,t_i}_{\mathsf c} z^{i,t_i}_{\mathsf d})
        \}}_{\mathsf{L}_{\mathsf{cd}}} \\
        & \ket{
        \cup_{i=1}^m
        \{
        (y^{i,t_i}_{\mathsf b}y^{i,t_i}_{\mathsf c},z^i_{\mathsf b}z^i_{\mathsf c})
        \}
        }_{\mathsf{R}_{\mathsf{bc}}}.
    \end{split}
    \end{equation}
    Each chain $i = 1,\ldots,m$ of length $t_i$ involves $t_i-1$ applications of the annihilation branch of $V_{\mathsf{bc}}$.
    Each annihilation undoes the action of a previous creation by $W_{\mathsf{bc},R}$, and thereby ``chains'' together the output of the $V_{\mathsf{ab}} \otimes V_{\mathsf{cd}}$ that preceded that $W_{\mathsf{bc},R}$ with the input of the $V_{\mathsf{ab}} \otimes V_{\mathsf{cd}}$ that proceeded the $V_{\mathsf{bc}}$.
    That is, the output of the earlier $V_{\mathsf{ab}} \otimes V_{\mathsf{cd}}$ is stored in pairs with the input of the later $V_{\mathsf{ab}} \otimes V_{\mathsf{cd}}$ in the $\mathsf{L}_{\mathsf{ab}}$ and $\mathsf{L}_{\mathsf{cd}}$ relation registers.
    We have $\sum_i t_i = t$, the total number of applications of the unitary.
    From our definitions, all of the $y^{i,\tau_i}_\mathsf{b}$ for any $i$ and any $0 \leq \tau_i \leq t_i$ are distinct, and similar for the $y^{i,\tau_i}_\mathsf{c}$.

    Meanwhile, the the states on the auxiliary registers when we apply $U'$ correspond to chains of the very similar form,
    \begin{equation} \label{eq: U' aux}
    \begin{split}
        & \ket{
        \cup_{i=1}^m
        \{
        (x^i_{\mathsf b}x^i_{\mathsf c},y^{i,0}_{\mathsf b}y^{i,0}_{\mathsf c})
        \}
        }_{\mathsf{L}_{\mathsf{bc}}} \\
        & \ket{
        \cup_{i=1}^m
        \{
        (x^{i,1}_{\mathsf a}y^{i,0}_{\mathsf b} y^{i,0}_{\mathsf c} x^{i,1}_{\mathsf d},
        z^{i,1}_{\mathsf a}y^{i,1}_{\mathsf b} y^{i,1}_{\mathsf c} z^{i,1}_{\mathsf d}),
        \ldots,
        (x^{i,t_i}_{\mathsf a} y^{i,t_i-1}_{\mathsf b} y^{i,t_i-1}_{\mathsf c} x^{i,t_i}_{\mathsf d},
        z^{i,t_i}_{\mathsf a} y^{i,t_i}_{\mathsf b} y^{i,t_i}_{\mathsf c} z^{i,t_i}_{\mathsf d})
        \}}_{\mathsf{L}_{\mathsf{abcd}}} \\
        & \ket{
        \cup_{i=1}^m
        \{
        (y^{i,t_i}_{\mathsf b}y^{i,t_i}_{\mathsf c},z^i_{\mathsf b}z^i_{\mathsf c})
        \}
        }_{\mathsf{R}_{\mathsf{bc}}},
    \end{split}
    \end{equation}
    where similar to before all of the $y^{i,\tau_i}_\mathsf{b}$ are distinct and all of the $y^{i,\tau_i}_\mathsf{c}$ are distinct.
    Moreover, the coefficients of each auxiliary register as above are identical to those of the auxiliary registers for $U$.
    That is, the only difference between the path-recording state in the experiment involving $U$ and that in the experiment involving $U'$ is the replacement of each auxiliary state in Eq.~(\ref{eq: U aux}) with the corresponding auxiliary state in Eq.~(\ref{eq: U' aux}).

    It remains only to show that there is a one-to-one isometry between the auxiliary states in Eq.~(\ref{eq: U aux}) and those in Eq.~(\ref{eq: U' aux}).
    Note that this is not a priori guaranteed, owing to the symmetrization of each set of pairs in each register.
    To show that there exists an isometry, we must show that for every state in Eq.~(\ref{eq: U' aux}) there is a unique associated state in Eq.~(\ref{eq: U aux}).
    The reverse direction follows trivially, since there is less symmetrization in Eq.~(\ref{eq: U aux}) than in Eq.~(\ref{eq: U' aux}). 
    This claim follows immediately from the fact that all $y^{i,\tau_i}_\mathsf{b}$ are distinct and all $y^{i,\tau_i}_\mathsf{c}$ are distinct. 
    %
    %
    Hence, for each value of $z^i_{\mathsf{b}}z^i_{\mathsf{c}}$ on the $\mathsf{R}_{\mathsf{bc}}$ register, there are a unique two values of $y^{i,t_i}_{\mathsf{b}}$ and $y^{i,t_i}_{\mathsf{c}}$, and in turn a unique two values of $y^{i,t_i-1}_{\mathsf{b}}$ and $y^{i,t_i-1}_{\mathsf{c}}$, and so on to $y^{i,0}_{\mathsf{b}}$ and $y^{i,0}_{\mathsf{c}}$, on the $\mathsf{L}_{\mathsf{ab}}$ and $\mathsf{L}_{\mathsf{cd}}$ registers, and thus a unique value of $y^{i,0}_\mathsf{b} y^{i,0}_\mathsf{c}$ and $x^i_\mathsf{b} x^i_\mathsf{c}$ on the $\mathsf{L}_{\mathsf{bc}}$ register.
    The local distinctness of $y^{i,\tau_i}_\mathsf{b}$ and $y^{i,\tau_i}_\mathsf{c}$ guarantees that for each $y^{i,\tau_i+1}_\mathsf{b}$ and $y^{i,\tau_i+1}_\mathsf{c}$ there is a unique $y^{i,\tau_i}_\mathsf{b}$ and $y^{i,\tau_i}_\mathsf{c}$ on $\mathsf{L}_{\mathsf{ab}}$ and $\mathsf{L}_{\mathsf{cd}}$.
    The distinctness of $y^{i,\tau_i}_\mathsf{b} y^{i,\tau_i}_\mathsf{c}$ from $y^{i,0}_\mathsf{b} y^{i,0}_\mathsf{c}$ guarantees that the step in process above where one jumps from the $\mathsf{L}_{\mathsf{ab}}$, $\mathsf{L}_{\mathsf{cd}}$ registers to the $\mathsf{L}_{\mathsf{bc}}$ register is uniquely determined.
    This completes our proof. \qed

\subsection{Proof of Proposition~\ref{prop:temp}: Impossibility of unitary designs from efficient temporal ensembles}

We consider the state $\Phi_\mathcal{E}(\dyad{0^n}^{\otimes k})$, in which the twirl over $k$ copies of a unitary drawn from $\mathcal{E}$ is applied to $k$ copies of the zero state on $n$ qubits.
For a Haar-random unitary, the state $\Phi_H(\dyad{0^n}^{\otimes k})$ has a flat spectrum across the symmetric subspace, which contains ${2^n + k - 1 \choose k}$ elements~\cite{schuster2024random}.
For any approximate unitary $k$-design with additive error $\varepsilon$, $\Phi_\mathcal{E}(\dyad{0^n}^{\otimes k})$ must be $\varepsilon$-close in trace distance to this state.

Let us now analyze the rank of the state $\Phi_\mathcal{E}(\dyad{0^n}^{\otimes k})$. We begin by considering the case of a single fixed Hamiltonian; our results will trivially extend to a bounded number $L$ of random Hamiltonians. 
Let us discretize the time interval $[0,T]$ into steps $\tau,2\tau,\ldots,(T/\tau)\tau$ of a small fixed size $\tau$.
For any probability distribution $\mathcal{D}$ over the evolution time $t \in [0,T]$, we can decompose the state $\Phi_\mathcal{E}(\dyad{0^n}^{\otimes k})$ into a mixture of contributions from nearby each time step,
\begin{equation} \label{eq:decompose PhiE EPR}
    \Phi_\mathcal{E}(\dyad{0^n}^{\otimes k})
    =
    \sum_{\ell=1}^{T/\tau} \mathcal{D}(\ell)
    \int_{(\ell-1) \tau}^{\ell \tau} dt  \mathcal{D}(t | \ell) \left( e^{-iH t} \dyad{0^n} e^{i H t} \right)^{\otimes k},
\end{equation}
where $\mathcal{D}(\ell)$ is the total probability to select a time in the $\ell$-th step, and $\mathcal{D}(t | \ell) = \mathcal{D}(t)/\mathcal{D}(\ell)$ is the conditional probability to select a time $t$ within the step $\ell$.
We can then apply the general bound,
\begin{equation} \nonumber
\begin{split}
	& \left\lVert 
	e^{-iH t} \dyad{0^n} e^{i H t} 
	- 
	e^{-iH \ell \tau} \dyad{0^n} e^{i H \ell \tau} 
    \right\rVert_1 \\
    & \quad \quad \leq 
    2 
    \left\lVert 
	(e^{-iH t} - 
	e^{-iH \ell \tau})
    \ket{0^n} 
    \right\rVert_2   \leq 
    2 
    \left\lVert 
	(e^{-iH (t-\ell \tau)} - 
	1)
    \ket{0^n} 
    \right\rVert_2  \leq 
    2 \lVert H \rVert_\infty  |t-\ell \tau|. \\
\end{split}
\end{equation}
where the first inequality follows from the inequality $\lVert \dyad{u} - \dyad{v} \rVert_1 \leq 2 \lVert \ket{u} - \ket{v} \rVert_2$~\cite{ma2025construct}.
In fact, we will need instead an identical version of this bound for the $k$-th power of both states,
\begin{equation} \nonumber
\begin{split}
	& \left\lVert 
	\left( e^{-iH t} \dyad{0^n} e^{i H t} \right)^{\otimes k} 
	- 
	\left( e^{-iH \ell \tau} \dyad{0^n} e^{i H \ell \tau} \right)^{\otimes k} 
    \right\rVert_1 \\
    & \quad \quad \quad \quad\quad \quad\quad \quad\quad \quad \leq 
    2 
    \left\lVert 
	((e^{-iH t})^{\otimes k} - 
	(e^{-iH \ell \tau})^{\otimes k})
    \ket{0^n}^{\otimes k} 
    \right\rVert_2   \leq 
    2 k \lVert H \rVert_\infty  |t-\ell \tau|. \\
\end{split}
\end{equation}
Inserting this bound into Eq.~(\ref{eq:decompose PhiE EPR}) and applying the triangle inequality to the integral over $t$, we find
\begin{equation} 
    \left\lVert \Phi_\mathcal{E}(\dyad{0^n}^{\otimes k})
    -
    \sum_{\ell=1}^{T/\tau} \mathcal{D}(\ell)
    \cdot  \left( e^{-iH \ell \tau} \dyad{0^n} e^{i H \ell \tau} \right)^{\otimes k} \right\rVert_1 \leq 2k \lVert H \rVert_\infty \tau.
\end{equation}
This shows $\Phi_\mathcal{E}(\dyad{0^n}^{\otimes k})$ can be approximated by a state with rank $T/\tau$ up to a small trace-norm error $2k\lVert H\rVert_\infty \tau$.

The proof is now complete. 
The state $\Phi_\mathcal{E}(\dyad{0^n}^{\otimes k})$ is $2k\lVert H \rVert_\infty \tau$-close to a state with rank $T/\tau$.
The state $\Phi_H(\dyad{0^n}^{\otimes k})$ has a flat spectrum with rank ${2^n + k - 1 \choose k}$.
A simple computation shows that any rank $r$ quantum state can be at most $2(1-r/r_0)$-close to a state with a flat spectrum of rank $r_0 > r$.
Hence, for the two states of interest to be $\varepsilon$-close, we must have
\begin{equation}
     2 \left(1-\frac{T}{\tau} {2^n+k-1 \choose k}^{-1} \right)
     -
     2k \lVert H \rVert_\infty \tau
    \leq \varepsilon.
\end{equation}
Let $\varepsilon = 1/4$ and set $\tau$ such that $2k \lVert H \rVert_\infty \tau = 1/4$ as well.
Hence, we require
\begin{equation}
     1-\frac{T}{\tau} {2^n+k-1 \choose k}^{-1}
    \leq \frac{1}{4},
\end{equation}
which requires
\begin{equation}
     T \geq \frac{3}{4} \tau {2^n+k-1 \choose k} = \frac{3}{32 k \lVert H \rVert_\infty} {2^n+k-1 \choose k} \geq \frac{3}{32 k \lVert H \rVert_\infty} \frac{2^{nk}}{k!}.
\end{equation}
If the ensemble involves $N_H$ Hamiltonians, the lower bound decreases by a factor of $N_H$.
This completes the proof. \qed






\end{document}